\documentclass[reqno]{amsart}
\usepackage[foot]{amsaddr}
\usepackage[margin=3cm]{geometry}
\usepackage{amsmath, amssymb,amsthm}
\usepackage[shortlabels]{enumitem}
\usepackage{mlmodern}
\DeclareSymbolFont{largesymbols}{OMX}{cmex}{m}{n} 
\usepackage{mathtools} 
\usepackage{stmaryrd} 
\usepackage{xcolor}
\usepackage{upref} 
\usepackage[colorlinks]{hyperref}
\hypersetup{citecolor=blue,filecolor=blue,linkcolor=blue,urlcolor=navyblue}
\definecolor{navyblue}{rgb}{0.0, 0.0, 0.5}

\newtheorem{thm}{Theorem}[section]
\newtheorem{prop}[thm]{Proposition}
\newtheorem{lem}[thm]{Lemma}
\newtheorem{cor}[thm]{Corollary}
\theoremstyle{definition}

\newtheorem{rmk}{Remark}[section]

\newtheorem*{claim*}{Claim}

\newcommand{\N}{\mathbb{N}}
\newcommand{\Z}{\mathbb{Z}}
\newcommand{\R}{\mathbb{R}}
\newcommand{\C}{\mathbb{C}}

\renewcommand{\Im}{\operatorname{\mathrm{Im}}}
\renewcommand{\Re}{\operatorname{\mathrm{Re}}}

\newcommand{\oneb}{\mathbf{1}}

\newcommand{\Tr}{\operatorname{Tr}}

\newcommand{\hs}{\mathrm{hs}}
\newcommand{\id}{\operatorname{id}}

\newcommand{\U}{\mathrm{U}}

\newcommand{\op}[1]{\operatorname{#1}}

\newcommand{\rank}{\op{rank}}
\newcommand{\supp}{\op{supp}}

\newcommand{\Hmin}{H_\mathrm{min}}
\newcommand{\Hpmin}{H_{p,\mathrm{min}}}
\newcommand{\SL}{\mathrm{SL}}
\newcommand{\s}{\mathbf s}
\renewcommand{\t}{\mathbf t}
\newcommand{\gap}{\mathfrak g}
\newcommand{\free}{\mathrm{free}}

\newcommand\numberthis{\stepcounter{equation}\tag{\theequation}}
\numberwithin{equation}{section}

\begin{document}

\title[]{A constructive violation of additivity of minimum output von Neumann entropy}

\author{Laura Shou$^{1,2}$}
\author{Alexey V. Gorshkov$^{1,2}$}

\address{\normalfont$^1$Joint Quantum Institute, Department of Physics, NIST/University of Maryland, College Park, MD 20742, USA}
\address{\normalfont$^2$Joint Center for Quantum Information and Computer Science,
NIST/University of Maryland, College Park, MD, 20742, USA}

\begin{abstract}
We give an explicit non-random example of nonadditivity of minimum output von Neumann entropy. The proof uses a finite-dimensional construction which imitates free Haar unitary behavior. The same channel gives a violation of additivity for the minimum output R\'enyi-$p$ entropy for any $p\in[1,\infty]$. Additionally, given any $p_0>0$ and $\gap>0$, we extend the construction to obtain an explicit non-random channel which has entropy gap $\ge\gap$ uniformly over $p\in[p_0,\infty]$.
\end{abstract}

\maketitle

\setcounter{tocdepth}{1}
\tableofcontents

\section{Introduction}

For a quantum channel $\Phi$, its minimum output entropy (MOE) is 
\begin{align}
\Hmin(\Phi):=\min_\rho H(\Phi(\rho)),
\end{align}
for\footnote{All logarithms are natural logarithms.} $H(\rho):=-\Tr(\rho\log\rho)$ the von Neumann entropy of a density matrix $\rho$.
While one always has $\Hmin(\Phi_1\otimes\Phi_2)\le\Hmin(\Phi_1)+\Hmin(\Phi_2)$ \cite{shor2004equivalence}, it has been of significant interest to construct channels which violate additivity, i.e. which satisfy
\begin{align}
\Hmin(\Phi_1\otimes\Phi_2)<\Hmin(\Phi_1)+\Hmin(\Phi_2),
\end{align}
meaning that an entangled input can produce lower von Neumann entropy than any product input.
Similarly, one can also ask the question for the minimum output R\'enyi-$p$ entropies, defined as $\Hpmin(\Phi):=\min_\rho H_p(\Phi(\rho))$ for $H_p(\sigma):=\frac{1}{1-p}\log\Tr(\sigma^p)$ for $p\in(0,\infty)$, with limits taken for the von Neumann entropy $p=1$ and for $p=0$ and $\infty$.

The first counterexample to additivity of minimum output R\'enyi-$p$ entropy was an explicit construction given for $p>4.79$ in \cite{werner2002counterexample}.
Some other early counterexamples were random, such as ones given for $p>1$ in \cite{hayden2008counterexamples}, and for $p$ near 0 in \cite{cubitt2008counterexamples}. 
The additivity question for the von Neumann entropy at $p=1$ remained of particular interest due to its equivalence to the question of additivity of Holevo capacity \cite{shor2004equivalence}.
Additivity for the minimum output von Neumann entropy was disproved by Hastings in \cite{hastings2009superadditivity}, using random Haar unitaries to construct counterexamples.
Since then, there has been much work concerning random counterexamples, including the introduction of free probability methods to study the minimum output entropy problem \cite{collins2011random,belinschi2012eigenvectors,belinschi2016almost}.
Many recent works continue to utilize free probability methods, including for establishing random counterexamples for general values of $p$, including $0\le p<1/4$ and $p>3/4$ in \cite{leung2026counterexamples}, and for all $p\ge0$ in \cite{zhen2026almost}.

Despite the abundance of random constructions of violations of MOE additivity, it has proven more difficult to write down explicit non-random examples.
Explicit constructions for $p>2$ were given in \cite{grudka2010constructive}, and extended to all $p>1$ in \cite{derksen2025constructive}.
However, explicit constructions for the von Neumann entropy case $p=1$ remained elusive. Recent works \cite{lovitz2026superadditivity,zhen2026deterministic} demonstrated deterministic polynomial-time algorithms for constructing violations of minimum output von Neumann entropy, using the derandomization of \cite{odonnell2020explicit}, but left open the problem of constructing a specific closed-form example.
In this paper, we resolve this problem by constructing an explicit non-random example of nonadditivity of minimum output von Neumann entropy. 
The proof, while constructing a non-random example, makes use of MOE developments based on free probability, particularly \cite{collins2018haagerup}, which demonstrated how Haagerup's inequality \cite{haagerup1978example} plays a key role in the MOE additivity problem. By imitating certain free properties in a non-random finite-dimensional channel, we will be able to construct a desired non-random counterexample.
Our main result is as follows.

\begin{thm}\label{thm:main}
Let $k=2^{32}$ and $d=3\cdot 2^{5\,268\,390\,106\,224\,805\,281\,791\,998}$. 
There are explicit unitaries $U_1,\ldots,U_k\in \U(d)$ and a hermitian $F\in M_d(\C)$, all written in Section~\ref{sec:construct}, such that the channel $\Phi:M_d(\C)\to M_k(\C)$ defined by
\begin{align}\label{eqn:phi}
[\Phi(\rho)]_{ij}&:=\frac1k\Tr(U_iF\rho FU_j^*)+\frac{\delta_{ij}}{k}\Tr((I_d-F^2)\rho),\quad i,j=1,\ldots,k
\end{align}
satisfies
\begin{align}\label{eqn:gap}
2\Hmin(\Phi)-\Hmin(\Phi\otimes\Phi)&>\frac{3}{17\,179\,869\,184}.
\end{align}
\end{thm}
\begin{rmk}
\begin{enumerate}[(i)]
\item Due to the dimension $d$, the construction would require around 5 septillion input qubits. 
We do not attempt to reduce the Hilbert space dimensions or increase the entropy gap in Theorem~\ref{thm:main}; instead we try to keep the proof as simple as possible.

\item The proof works by imitating certain aspects of free Haar unitary behavior (see Section~\ref{sec:prelim}), without using strong convergence. Since many random constructions in the MOE literature make use of freeness and strong convergence, the methods here may be applicable to other constructions as well.

\end{enumerate}

\end{rmk}

As a corollary of some of the ingredients used to prove Theorem~\ref{thm:main}, we also obtain
\begin{cor}\label{cor:p}
The channel $\Phi$ defined in \eqref{eqn:phi} violates additivity of minimum output R\'enyi-$p$ entropy for all $1\le p\le\infty$, with a uniform lower bound $\frac{3}{17\,179\,869\,184}$ on the gap.
\end{cor}

The main idea of the proof of Theorem~\ref{thm:main}, which is explained further in Section~\ref{subsec:proof-idea}, is to construct unitary (permutation) matrices $U_1,\ldots,U_k$ which behave like a freely independent Haar unitary family up to a fixed large trace moment. For free Haar unitary behavior, \cite{collins2018haagerup} showed how Haagerup's inequality \cite{haagerup1978example} can be used to obtain a violation of additivity of minimum output entropy. 
The construction of \cite{collins2018haagerup} uses strong convergence of a random matrix family to a free (infinite-dimensional) family as $n\to\infty$, and the resulting examples of nonadditivity are random.
In the non-random construction here, we do not use or have strong convergence to the free family, but we will still imitate the free behavior by matching normalized trace moments with those of free Haar unitaries up to a fixed finite order.
This will however not fully control non-free behavior, and so a damping term $F$ will be used to suppress the undesired behavior. Overall, the moment agreement and $F$ will be enough to ensure the behavior is close enough to free to obtain a violation of additivity.

\vspace{2mm}
Finally, while we aim to keep the proof of Theorem~\ref{thm:main} as simple as possible, we note that for any $\gap>0$ and $p_0>0$, we can use a slightly more complicated construction and estimates to construct an explicit channel which violates additivity of minimum output R\'enyi-$p$ entropy with an entropy gap $\ge\gap$ uniformly over all $p\ge p_0>0$:
\begin{thm}\label{thm:ext}
Fix $p_0>0$ and $\ell\in\N$.
Let $N=2^{17\ell}$ and $D=(3\cdot 2^{336m(\ell,p_0)-2})^{16\ell}$, where 
\begin{align}
m(\ell,p_0)=4+\left\lceil\max\left\{\frac{\ell+7}{\min(p_0,1/2)},17\ell+7\right\}\right\rceil+(17\ell+5)2^{34\ell-1},
\end{align}
which is $\Omega(\ell(p_0^{-1}+2^{34\ell}))$.
Then there are explicit unitaries $W_1,\ldots,W_N\in\U(D)$ and a hermitian $F\in M_D(\C)$, defined in Section~\ref{subsec:ext-construction}, such that the channel $\Psi:M_D(\C)\to M_N(\C)$ defined by
\begin{align}\label{eqn:phi2}
[\Psi(\rho)]_{\s\t}&=\frac1N\Tr(W_\s F\rho FW_\t^*)+\frac{\delta_{st}}{N}\Tr((I-F^2)\rho),
\end{align}
satisfies for every $p\in[p_0,\infty]$,
\begin{align}
2\Hpmin(\Psi)-\Hpmin(\Psi\otimes\Psi)>\left(\frac45\ell-6\right)\log2.
\end{align}
\end{thm}

The proof of Theorem~\ref{thm:ext} combines the explicit construction method of Theorem~\ref{thm:main} with the large entropy gap methods of \cite{wang2026unbounded,zhen2026almost}.
To enlarge the entropy gap, the idea is to consider (commuting) products of the free group $F_k$ on $k$ generators, and then use a product form of Haagerup's inequality due to \cite{collins2022additivity}.

\subsection{Outline}
The rest of this paper is organized as follows. 
In Section~\ref{sec:prelim}, we review some background, give an overview of the proof ideas for the main theorem, and provide some derivations and motivations for the later constructions.
In Section~\ref{sec:construct}, we give the explicit constructions of the $U_j$ and $F$, and prove several useful properties about them.
In Section~\ref{sec:end}, we complete the proof of Theorem~\ref{thm:main}.
The appendices provide the proofs of Corollary~\ref{cor:p} and Theorem~\ref{thm:ext}.

\section{Preliminaries and proof overview}\label{sec:prelim}

For Hilbert space dimensions $d,k$, recall a density matrix $\rho\in M_d(\C)$ is a positive semidefinite matrix ($\rho\ge0$) with trace one, and a quantum channel $\Phi:M_d(\C)\to M_k(\C)$ is a linear map which is completely positive and trace-preserving, i.e. $\Phi\otimes \id_s$ maps positive semidefinite matrices to positive semidefinite matrices for any auxiliary dimension $s$, and $\Tr\Phi(\rho)=\Tr\rho$.

\subsection{Freeness} \label{subsec:free}
We review some definitions and background on free Haar unitaries; for a much more complete treatment, see the book \cite{nica2006lectures}. 
First, recall a $*$-probability space is a pair $(\mathcal A,\tau)$, where
\begin{itemize}
\item $\mathcal A$ is a unital algebra over $\C$, with antilinear $*$-operation $a\mapsto a^*\in\mathcal A$ such that $(a^*)^*=a$ and $(ab)^*=b^*a^*$.
\item $\tau$ is a linear functional $\tau:\mathcal A\to\C$ satisfying $\tau(1)=1$ and $\tau(a^*a)\ge0$ for all $a\in\mathcal A$. If $\tau(ab)=\tau(ba)$, then $\tau$ is called a trace.
\end{itemize}
Unitary elements are those $u\in\mathcal A$ satisfying $u^*u=uu^*=1$.
A \emph{Haar unitary} is a unitary satisfying
\begin{align}
\tau(u^j)=0,\quad\forall\, j\in\Z\setminus\{0\}.
\end{align}
(This is not used to mean a Haar-random unitary matrix; these unitaries need not be random nor matrices.)
For a $*$-probability space $(\mathcal A,\tau)$, elements $a_j\in\mathcal A$, $j=1,\ldots,k$, are called \emph{$*$-freely independent} if 
\begin{align}
\tau(p_1(a_{i_1})\cdots p_m(a_{i_m}))=0,
\end{align}
for every choice of noncommutative polynomials $p_1,\ldots,p_m\in\C\langle X,X^*\rangle$ satisfying $\tau(p_j(a_{i_j}))=0$ for all $j=1,\ldots,m$, and any sequence of indices with $i_j\ne i_{j+1}$ (differing consecutive indices). 
We will also refer to this as freely independent, with the understanding that we mean $*$-freely independent.

The standard way to construct an abstract family of freely independent Haar unitaries $u_1,\ldots u_k$ is as follows. Let $F_k$ be the free group on $k$ generators $g_1,\ldots,g_k$, and let $\mathcal B(\ell^2(F_k))$ denote the space of bounded linear operators on $\ell^2(F_k)$, and $\mathcal U(\ell^2(F_k))$ the space of unitary operators on $\ell^2(F_k)$. Let $\lambda:F_k\to \mathcal U(\ell^2(F_k))$ be the left-regular representation of the (discrete) group $F_k$ defined by
\begin{align}\label{eqn:lr}
\lambda(g)|h\rangle=|gh\rangle,\quad\text{for }g,h\in F_k,
\end{align}
and $\{|h\rangle:h\in F_k\}$ the standard orthonormal basis of $\ell^2(F_k)$. Define the unitaries $u_j\in\mathcal U(\ell^2(F_k))$, $j=1,\ldots,k$, and trace $\tau$ on the algebra generated by $\{\lambda(g)\}_{g\in F_k}$ via
\begin{align}\label{eqn:uj}
u_j:=\lambda(g_j),\;j=1,\ldots,k,\quad \tau(x):=\langle e|x|e\rangle, 
\end{align}
for $|e\rangle$ the basis function corresponding to the identity $e\in F_k$.
Then $\tau(u_j^\ell)=0$ for all $\ell\in\Z\setminus\{0\}$ since $g_j^\ell\ne e$, so $u_j$ are Haar unitaries. 
The unitaries $u_1,\ldots u_k\in\mathcal U(\ell^2(F_k))$ then form a freely independent Haar unitary family by the definition of $\lambda$ and free group $F_k$.
We use ``free Haar unitaries'' to mean freely independent Haar unitaries.

An important result in operator algebras and free probability is Haagerup's inequality \cite{haagerup1978example}. It was shown in \cite{collins2018haagerup} how this also plays a key role in the MOE additivity problem. 
Later works \cite{collins2022additivity,fukuda2022additivity,kalantar2025more,zhen2026deterministic,wang2026unbounded,zhen2026almost} also made use of this or extensions of Haagerup's inequality to further construct and study violations of additivity.

\begin{thm}[Haagerup's inequality, {\cite[Lemma 1.4]{haagerup1978example}}]\label{thm:haagerup}
Let $F_k$ be the free group on $k$ generators, and let $f$ be a function in $\ell^2(F_k)$ supported on reduced words of length $n$ in $F_k$. Then letting $\lambda(f):=\sum_{g\in F_k}f(g)\lambda(g)$, we have the operator norm bound
\begin{align}\label{eqn:haagerup}
\|\lambda(f)\|&\le(n+1)\|f\|_2.
\end{align}
\end{thm}
In \cite{collins2018haagerup}, Collins showed how to use the $n=2$ case of Theorem~\ref{thm:haagerup} to prove the channel $\tilde\Phi:M_n(\C)\to M_k(\C)$ defined by
\begin{align}\label{eqn:phi-collins}
[\tilde\Phi(\rho)]_{ij}=\frac1k\Tr(U_i^{(n)}\rho\,(U_j^{(n)})^*),
\end{align}
for $n\times n$ independent Haar distributed random unitaries $U_1^{(n)},\ldots, U_k^{(n)}$, gives a violation of additivity of MOE with probability 1 as $n\to\infty$ and for sufficiently large fixed $k$.
(This is also related to the complementary channel of the original Hastings \cite{hastings2009superadditivity} Haar random channel.)
The proof in \cite{collins2018haagerup} uses strong convergence of the unitaries $U_1^{(n)},\ldots,U_k^{(n)}$ to freely independent Haar unitaries \cite{collins2014strong} in order to invoke Theorem~\ref{thm:haagerup} with $n=2$. Note that the channel $\Phi$ in Theorem~\ref{thm:main} is of a similar form as \eqref{eqn:phi-collins}; if the damping factor $F$ were not needed and were just the identity, then $\Phi$ reduces to the channel \eqref{eqn:phi-collins}. 

As in \cite{collins2018haagerup}, we will only need the length $n=2$ case of Theorem~\ref{thm:haagerup}. 
For a $k \times k$ matrix $A$ with zero diagonal, define the function $f_A$, supported on length-2 words, by $f_A(g_i^{-1} g_j) = A_{ij}$ for $i \neq j$, and $f_A(g) = 0$ otherwise.
Then $\|f_A\|_2^2=\|A\|_\hs^2$, for $\|\cdot\|_\hs$ the Hilbert--Schmidt or Frobenius norm, and we can define
\begin{align}\label{eqn:TA}
T_{A,\mathrm{free}}:=\lambda(f_A)=\sum_{i\ne j}A_{ij}\lambda(g_i)^*\lambda(g_j).
\end{align}
Haagerup's inequality Theorem~\ref{thm:haagerup} for $n=2$ then becomes
\begin{align}\label{eqn:haagerup2}
\|T_{A,\mathrm{free}}\|&\le 3\|A\|_\hs.
\end{align}

\subsection{Proof overview}\label{subsec:proof-idea}

To prove non-additivity for $\Phi$ in \eqref{eqn:phi}, we need to prove (1) $H(\Phi(\rho))$ is ``large'' for all possible inputs $\rho$, and (2) there is some entangled state $\rho_{12}$ for which $H((\Phi\otimes\Phi)(\rho_{12}))$ is ``small''. The latter in finite-dimensional systems is typically the easier direction, since one can often use a Bell state $\rho_{12}=|\Omega_d\rangle\langle\Omega_d|$ for $|\Omega_d\rangle=d^{-1/2}\sum_{j=1}^d|j\rangle\otimes|j\rangle$. The Bell state will work here as well, although we will have to show the damping term $F$ in \eqref{eqn:phi} does not cause an issue.
For the former, as discussed above, \cite{collins2018haagerup} showed that if one had free Haar unitaries $u_j$, or more precisely, finite-dimensional matrices $U_j^{(n)}$ converging strongly to free Haar unitaries $u_j$, then one could define the channel \eqref{eqn:phi-collins} and use Haagerup's inequality \cite{haagerup1978example} to obtain the required largeness of $H(\Phi(\rho))$ and violation of additivity.
The construction of \cite{collins2018haagerup} is however random, using strong convergence to (infinite-dimensional) free unitaries, and it is still an open question in the area to construct explicit deterministic sequences that converge strongly to the appropriate free unitary family \cite{magee2025strong,vanhandel2026strong}. Fortunately, we do not need full strong convergence, and it will be enough to match some free moments.

More specifically, we use a family of finite-dimensional unitary matrices whose normalized trace moments agree with those of free Haar unitaries up to a fixed finite order.
The moments do not control the unitaries' behavior on all subspaces, so we also add the damping factor $F$ to suppress undesirable behavior.
The first term $\frac1k\Tr(U_iF\rho FU_j^*)$ in \eqref{eqn:phi} is the analogue of the random channel $\tilde\Phi$ with the damping factor $F$ installed. Since $F$ may suppress some of the trace, the second term $\Tr((I-F^2)\rho)I_k/k$ in \eqref{eqn:phi} is added solely to make $\Phi$ trace-preserving.
We will then show that, due to the free moment behavior of the unitaries, we can take $F$ similar enough to the identity, so that the channel $\Phi$ in \eqref{eqn:phi} retains enough similarities to the random channel example $[\tilde\Phi(\rho)]_{ij}= \frac{1}{k} \Tr(U_i^{(n)}\rho\,(U_j^{(n)})^*)$, to similarly obtain a violation of MOE additivity.

\subsection{Minimum output entropy and \texorpdfstring{$T_A$}{TA}}\label{subsec:prelim-ta}
To motivate the constructions in Section~\ref{sec:construct}, we follow \cite{collins2018haagerup} and explain how the free probability quantity $T_{A,\mathrm{free}}$ in \eqref{eqn:TA}, or more specifically its finite-dimensional analogue 
\begin{align}\label{eqn:TA-finite0}
T_A:=\sum_{i\ne j}A_{ij}U_i^*U_j,
\end{align}
for $U_i$ the unitaries in Theorem~\ref{thm:main} and defined in Section~\ref{subsec:construction}, relates to $\Hmin(\Phi)$.
To show that $H(\Phi(\rho))$ is large for all input density matrices $\rho$, we want to show that $\Phi(\rho)$ is close to the maximally mixed state $\frac1kI_k$ for any $\rho$, in particular that the Hilbert--Schmidt or Frobenius norm $\|\Phi(\rho)-I_k/k\|_\hs$ is small. Due to the entropy estimate\footnote{This can also be proved directly using the inequality $\log x\le x-1$.} \cite{hastings2009superadditivity,collins2018haagerup}
\begin{align}\label{eqn:Hbound}
H(\sigma)&\ge\log k-k\|\sigma-I_k/k\|_\hs^2,
\end{align}
controlling $\|\Phi(\rho)-I_k/k\|_\hs$ will give a quantitative lower bound on $H(\Phi(\rho))$, which can be used to show it is large.

The quantity $\|\Phi(\rho)-I_k/k\|_\hs$ can be expressed in terms of $T_A$ as follows. 
Let $D:=\Phi(\rho)-\frac1kI_k$, and observe that for self-adjoint, zero-diagonal $D$,
\begin{align}\label{eqn:D}
\|D\|_\hs&=\max_{\|A\|_\hs=1,A=A^*,\operatorname{diag}(A)=0}\Tr[AD],
\end{align}
with the maximum obtained for $A=D/\|D\|_\hs$ (for $D\ne0$). 
The channel $\Phi$ has diagonal entries all $1/k$, so $D$ has zero diagonal and we can indeed optimize only over $A$ with $\operatorname{diag}A=0$. 
Since $A_{ii}=0$, a quick expansion with the definition of $\Phi$ in \eqref{eqn:phi} shows
\begin{align}\label{eqn:trad}
\Tr[AD]=\Tr[A(\Phi(\rho)-I_k/k)]=\frac1k\Tr(\rho FT_AF)\le \frac1k\|FT_AF\|.
\end{align}
If we were somehow in the truly free setting, then we would take $F=I$ and have $\|T_A\|\le 3$ by Haagerup's inequality, which would give a good bound on $\|D\|_\hs$. 
Instead, while the finite-dimensional analogue $T_A$ may behave like $T_{A,\mathrm{free}}$ on most subspaces, we will have to introduce the damping factor $F$ to suppress large operator norm of $T_A$ from exceptional subspaces. The quantity $\|FT_AF\|$ will imitate Haagerup's inequality by construction, and free-like moments of the $U_i$'s will show that $F$ is similar enough to the identity for the rest of the additivity argument to go through.

\section{Explicit construction}\label{sec:construct}

In this section, we define the matrices $U_i$ and $F$ used in Theorem~\ref{thm:main}, and prove several useful properties about them.
Explicit choices of parameters are made in Section~\ref{subsec:parameters}.

\subsection{Unitary construction}\label{subsec:construction}

To define the $U_i$'s, we first need an explicit representation of the free group $F_2$ on 2 generators. A particularly useful construction is the Sanov representation \cite{sanov1947property}, for which Sanov proved that the two $\SL_2(\Z)$ matrices
\begin{align}\label{eqn:sanov}
P=\begin{pmatrix}1&2\\0&1\end{pmatrix},\quad Q=\begin{pmatrix}1&0\\2&1\end{pmatrix},
\end{align}
generate an instance of $F_2$ in $\SL_2(\Z)$. Recall a subset $\{g_i\}_i\subseteq G$ is free if every nonempty reduced word is nontrivial, i.e. if every reduced word $g_{i_1}^{\varepsilon_1}\cdots g_{i_k}^{\varepsilon_k}\ne e$, for $i_j\ne i_{j+1}$, any $k\ge1$, and $\varepsilon_j\in\Z\setminus\{0\}$. 
Since $P$ and $Q$ are free generators, 
we see the matrices
\begin{align}\label{eqn:aj}
a_j:=P^jQP^{-j}=\begin{pmatrix}1+4j&-8j^2\\2&1-4j\end{pmatrix},
\end{align}
for any set of distinct $j\in\Z$, are free. The standard construction \eqref{eqn:uj} then forms an explicit freely independent Haar unitary family.

In order to produce finite-dimensional $U_j$'s which retain some of the free behavior, we will reduce the group generated by the $a_j$'s modulo a large integer $M=2^n$, and consider the group $G=\SL_2(\Z/M\Z)$. 
This is a specific instance of the finite-quotient constructions which give permutation-matrix approximations of free groups \cite{pestov2008hyperlinear}.
Fix an integer $k=2^s$ with $s\ge1$; we will later choose $k=2^{32}$ in Section~\ref{subsec:parameters}, in agreement with Theorem~\ref{thm:main}.
The resulting unitaries are defined in the same way as \eqref{eqn:uj}, via
\begin{align}\label{eqn:Uj-finite}
U_j|g\rangle:=|a_jg\rangle,\quad j=1,\ldots,k,
\end{align}
for $a_j$ as in \eqref{eqn:aj} reduced modulo $M$, and any $g\in G=\SL_2(\Z/M\Z)$.
The unitaries $U_j$ are permutation matrices and have size (see \cite[p.46]{apostol1990modular} or OEIS \href{https://oeis.org/A000056}{A000056}),
\begin{align}\label{eqn:d}
d=|\SL_2(\Z/2^n\Z)|=3\cdot 2^{3n-2}.
\end{align}
We will choose $M=2^n$ large enough so that every nontrivial reduced word in the $a_j^{\pm1}$'s of length $\le 4m$, for $m$ a parameter to be chosen later in Section~\ref{subsec:parameters}, remains nontrivial. Using the explicit form of the entries in \eqref{eqn:aj} for $j=1,\ldots,k$, which have easy inverses since they are in $\SL_2(\Z)$, we see the largest absolute row or column sum
of any $a_j$ or $a_j^{-1}$ is $\le 1+4k+8k^2<12k^2$. 
Under the matrix multiplication, the largest absolute row or column sum is submultiplicative, so a reduced word $w$ of length at most $4m$ has maximum entry at most $(12k^2)^{4m}$. We want to make sure $w\ne I_2$, so we consider $w-I_2$, which could increase the maximum entry by 1. Thus to ensure no nonempty reduced words $w$ of length at most $4m$ reduce to $I_2$ modulo $M$, we use
\begin{align}\label{eqn:M}
(12k^2)^{4m}+1<(16k^2)^{4m}=:M,
\end{align}
with the $16$ chosen since it is a power of 2.
For $k$ a power of $2$, we then have $M=2^n$ for $n=16m+8m\log_2k$.

The choice of large $M$ will ensure that trace moments of the finite-dimensional analogue $T_A$ defined in \eqref{eqn:TA-finite} match the free value for $T_{A,\mathrm{free}}$ from \eqref{eqn:TA}, up to order $2m$. 
We take the normalized trace $\tau_d(T):=\frac1d\Tr(T)$ for $T$ acting on $\SL_2(\Z/M\Z)$. Then for the left-regular representation $\lambda_G(g)|h\rangle=|gh\rangle$,
\begin{align}
\tau_d(\lambda_G(g))=\frac1d\sum_{h\in G}\langle h|gh\rangle=\oneb_{g=e}.
\end{align}
This agrees with the free case \eqref{eqn:uj}, $\tau_{\mathrm{free}}(\lambda_{F_k}(g))=\oneb_{g=e}$.

Analogously to \eqref{eqn:TA}, and in agreement with \eqref{eqn:TA-finite0}, we take the finite-dimensional
\begin{align}\label{eqn:TA-finite}
T_A:=\sum_{i\ne j}A_{ij}U_i^*U_j=\sum_{i\ne j}A_{ij}\lambda_G(a_i^{-1}a_j).
\end{align}
Thus each individual summand term in $T_A^{2m}$ corresponds to a word in $G=\SL_2(\Z/M\Z)$ of length at most $4m$. 
Since modding out by $M$ cannot make one such nontrivial word become trivial, we see that
\begin{align*}
\tau_d(T_A^{2m})
&=\sum_{i_1\ne j_1,\ldots,i_{2m}\ne j_{2m}}A_{i_1j_1}\cdots A_{i_{2m}j_{2m}}\tau_d(\lambda_G(a_{i_1}^{-1}a_{j_1}\cdots a_{i_{2m}}^{-1}a_{j_{2m}}))\\
&=\sum_{i_1\ne j_1,\ldots,i_{2m}\ne j_{2m}}A_{i_1j_1}\cdots A_{i_{2m}j_{2m}}\tau_{\mathrm{free}}(\lambda_{F_k}(a_{i_1}^{-1}a_{j_1}\cdots a_{i_{2m}}^{-1}a_{j_{2m}}))
=\tau_{\mathrm{free}}(T_{A,\mathrm{free}}^{2m}).\numberthis
\end{align*}
For $A$ self-adjoint, $T_A$ and $T_{A,\mathrm{free}}$ are self-adjoint. 
Then using Haagerup's inequality Theorem~\ref{thm:haagerup} for the free setting, we obtain 
\begin{lem}[moment matching]\label{lem:moments}
Let $A$ be a $k\times k$ hermitian matrix and $\|A\|_\hs=1$. Then for $T_A$ defined in \eqref{eqn:TA-finite} and $T_{A,\mathrm{free}}$ defined in \eqref{eqn:TA}, we have
\begin{align}\label{eqn:trace3}
\tau_d(T_A^{2m})=\tau_{\mathrm{free}}(T_{A,\mathrm{free}}^{2m})\le \|T_{A,\mathrm{free}}^{2m}\|\le 3^{2m}. 
\end{align}
\end{lem}
This quantifies the heuristic that $T_A$ cannot be much larger than 3 except perhaps on a small subspace.
This similarity to free behavior, along with the damping properties discussed next, will be used in Section~\ref{sec:end} to prove Theorem~\ref{thm:main}.

\subsection{Damping construction}\label{subsec:damping}

Even though the moments match up to a large value $2m$, the finite-dimensional $U_j$'s display very non-free behavior on certain subspaces, e.g.~all the $U_j$'s are finite permutation matrices and so preserve the vector $\sum_{g\in G}|g\rangle$, which can then produce a large norm for $T_A$; in particular, the associated sequence of families $(U_1^{(d)},\ldots,U_k^{(d)})$ does not converge strongly to a free Haar unitary family.
To resolve this problem, we define the damping operator $F$ which penalizes any subspace where $T_A$ could be large. By construction, this forces a Haagerup-like inequality for $\|FT_AF\|$. The free-like behavior of the $U_j$'s will then be used to show that $F$ can be taken close enough to the identity (in terms of normalized trace) to allow the rest of the MOE additivity violation argument, particularly the small $\Hmin(\Phi\otimes\Phi)$ using the Bell witness, to go through.

From \eqref{eqn:trad} in Section~\ref{subsec:prelim-ta}, we want to construct $F$ so that $\|FT_AF\|$ is small.
We don't know which specific $A$ we need to focus on, since it depends on $D=D(\rho)$ from \eqref{eqn:D}, so we effectively put a penalty on all of them where $T_A$ is large. To do this, we form an $\varepsilon$-net of the Frobenius sphere $S_F:=\{A:A=A^*,\|A\|_\hs=1,\operatorname{diag}(A)=0\}$, as this is the set of $A$'s to consider in \eqref{eqn:D}.

To make everything in the construction explicit, we give an explicit net as follows.
The standard real orthonormal basis for zero-diagonal complex hermitian matrices with the Frobenius inner product is given by the $r=k(k-1)$ matrices $E_{ab}^{\mathrm R}=\frac{1}{\sqrt{2}}(|a\rangle\langle b|+|b\rangle\langle a|)$ and $E_{ab}^{\mathrm I}=\frac{i}{\sqrt{2}}(|a\rangle\langle b|-|b\rangle\langle a|)$, $a<b$.
For an integer $\eta>\sqrt{r}$, choose the set
\begin{align}\label{eqn:net}
\mathcal N_\eta:=\left\{\frac{\sum_{a<b}(x_{ab}E_{ab}^{\mathrm R}+y_{ab}E_{ab}^{\mathrm I})}{\sqrt{\sum_{a<b}(x_{ab}^2+y_{ab}^2})}: (x_{ab},y_{ab})_{a<b}\in\{-\eta,\ldots,\eta\}^r\setminus\{0\}\right\}\subset S_F.
\end{align}
Note that $\sum_{a<b}(x_{ab}E_{ab}^{\mathrm R}+y_{ab}E_{ab}^{\mathrm I})$ is just the hermitian matrix whose entries are $\frac{1}{\sqrt{2}}(x_{ab}+iy_{ab})$ for $a<b$.
Given any unit $v\in S_F$, consider the hermitian matrix $W$ formed by rounding all coordinates of $\eta v$ to the nearest integer. 
For $\eta>\sqrt{r}$, at least one of the $r$ real coordinate of $\eta v$ has absolute value $\ge1$, so $W\ne0$.
Then define $w:=W/\|W\|_\hs\in\mathcal N_\eta$, and check
\begin{align*}
\|w-v\|_\hs&\le \|w-W/\eta\|_\hs+\|W/\eta-v\|_\hs=|1-\|W\|_\hs/\eta|+\|W/\eta-v\|_\hs\\
&\le 2\|W/\eta-v\|_\hs\le \frac{1}{\eta}\sqrt{r}. \numberthis
\end{align*}
Thus $\mathcal N_\eta$ is an $\varepsilon:=\sqrt{r}/\eta$-net for $S_F$ of size $|\mathcal N_\eta|\le (2\eta+1)^r$.
Letting $q:=(x_{ab},y_{ab})_{a<b}\in\{-\eta,\ldots,\eta\}^r\setminus\{0\}$, let $A_q$ be the point in the net
\begin{align}
A_q:=\frac{\sum_{a<b}(x_{ab}E_{ab}^{\mathrm R}+y_{ab}E_{ab}^{\mathrm I})}{\sqrt{\sum_{a<b}(x_{ab}^2+y_{ab}^2)}}.
\end{align}
For all $q\in\{-\eta,\ldots,\eta\}^r\setminus\{0\}$, we want to make sure $\|FT_{A_q}F\|$ is not too large. We do this by defining, for a constant $L>3$ to be chosen later,
\begin{align}\label{eqn:R}
R:=\sum_{q\in\{-\eta,\ldots,\eta\}^r\setminus\{0\}}\left(\frac{T_{A_q}}{L}\right)^{2m},\quad F:=(I+R)^{-1}.
\end{align}
Since $T_{A_q}$ are self-adjoint for hermitian $A_q$, we see $R\ge0$, which also implies $0< F\le1$.
Intuitively, if $T_{A_q}v$ had large norm, then $R$ would be large on $v$, and the damping factor $F$ would suppress the large-$R$ part of $v$. 
The parameter $m$ lets us tune the strength of the penalty, and will also be used to counteract the size of the net in the estimates below. Large $m$ will both amplify deviations of $T_{A_q}$ above $L$, and reduce the trace of $R$ which helps with making $F$ more similar to the identity. The trade-off is that larger $m$ will lead to larger $M$ in \eqref{eqn:M} and larger input dimension $d$.

From the moment matching Lemma~\ref{lem:moments}, we can estimate
\begin{align}
\tau_d(R)&\le (2\eta+1)^r\left(\frac{3}{L}\right)^{2m}.
\end{align}
To control $F$, note that for $x\ge0$, algebraic manipulation gives $(1+x)^{-2}\ge1-2x$, so
\begin{align*}
\tau_d(F^2)&=\tau_d((I+R)^{-2})\ge 1-2\tau_d(R)\ge 1-2(2\eta+1)^r\left(\frac3L\right)^{2m}.\numberthis\label{eqn:tauF2}
\end{align*}
We will later choose parameters $\eta,m,L$ so this is close to 1. Equation~\eqref{eqn:tauF2} will be used in Proposition~\ref{prop:2} to show the two-use channel $\Phi\otimes\Phi$ on the Bell state has low entropy.
In particular, a large trace in \eqref{eqn:tauF2} will rule out the possibility that e.g. $F\approx0$, which would make bounding $\|FT_AF\|$ easy, but would destroy the entropy gap argument since then $\Phi(\rho)\approx I_k/k$.

We can check that $R$ is real due to the symmetry $(x,y)\mapsto(x,-y)$ in the net $\mathcal N_\eta$, as $T_{A_{(x,y)}}^{2m}+T_{A_{(x,-y)}}^{2m}$ is real, since $A_{(x,-y)}=\overline{A_{(x,y)}}$ so that $T_{A_{(x,-y)}}=T_{\overline{A_{(x,y)}}}=\overline{T_{A_{(x,y)}}}$, since the $U_i$'s are real.
Thus $F$ is real, and $\Phi$ is also real.

We combine some of the above results, as well as a bound on $\|FT_{A_q}F\|$, into
\begin{lem}[$F$ properties]\label{lem:F}
The damping term $F$ defined in \eqref{eqn:R} is real symmetric, and satisfies $0<F\le1$ and the trace bound \eqref{eqn:tauF2}.
Additionally, for every $A_q$ in the net \eqref{eqn:net},
\begin{align}\label{eqn:lbound}
\|FT_{A_q}F\|\le L.
\end{align}
\end{lem}
\begin{proof}
We only still need to prove \eqref{eqn:lbound}, for which we apply some functional calculus inequalities.
First observe that for $t\in\R$, we have $|t|\le 1+t^{2m}$ for any $m\ge1$; thus from the definition of $R$ in \eqref{eqn:R}, we see
\begin{align}
-L(I_d+R)\le T_{A_q}\le L(I_d+R).
\end{align}
One can directly check that the map $X\mapsto FXF^*$, for any matrix $F$, preserves operator order, and so
\begin{align}
-LF\le FT_{A_q}F\le LF.
\end{align}
Since $0<F\le I$, we obtain $\|FT_{A_q}F\|\le L$.
\end{proof}

\section{Proof of Theorem~\ref{thm:main}}\label{sec:end}

In this section, we prove Theorem~\ref{thm:main}. Recall this requires us to show (1) $H(\Phi(\rho))$ is large for all $\rho$, and (2) there is some entangled state $\rho_{12}$ for which $H((\Phi\otimes\Phi)(\rho_{12}))$ is small.
We however start by proving
\begin{lem}\label{lem:channel}
$\Phi$ defined in \eqref{eqn:phi} is a channel.
\end{lem}
\begin{proof}
We need to verify $\Phi$ defined in \eqref{eqn:phi} is completely positive and trace-preserving.
First, the entire point of the second term in \eqref{eqn:phi} is to make $\Phi$ trace preserving, since $F$ can suppress some of the trace. Summing the diagonal entries in \eqref{eqn:phi}, we obtain
\begin{align}
\Tr\Phi(\rho)=\Tr(F^2\rho)+\Tr((I_d-F^2)\rho)=\Tr(\rho),
\end{align}
as desired.

To show $\Phi$ is completely positive, write it as $\Phi(\rho)=\Lambda(\rho)+\Gamma(\rho)$, for $[\Lambda(\rho)]_{ij}=\frac1k\Tr(U_iF\rho FU_j^*)$ the first part and $\Gamma(\rho)=\Tr((I-F^2)\rho)I_k/k$ the second part of \eqref{eqn:phi}. 
We will show both $\Lambda$ and $\Gamma$ are completely positive.
For $\Lambda$, define the map $W:\C^d\to \C^k\otimes\C^d$ below, and write $\Lambda$ in the Stinespring representation
\begin{align}
\Lambda(\rho)&=\Tr_{\C^d}(W\rho W^*),\quad\text{for}\quad W:=k^{-1/2}\sum_{i=1}^k|i\rangle\otimes U_iF.
\end{align}
This is completely positive by \cite[Theorem 2.22]{watrous2018book}.

For $\Gamma(\rho)=\Tr((I-F^2)\rho)I_k/k$, let $S:=(I-F^2)^{1/2}$, and consider any auxiliary dimension $s$ and $X\ge0$ in $M_d(\C)\otimes M_s(\C)$. Expanding $X$ in a tensor basis shows
\begin{align}
(\Gamma\otimes\operatorname{id}_s)(X)&=\frac1kI_k\otimes\Tr_{\C^d}((S\otimes I_s)X(S\otimes I_s)).
\end{align}
As noted before, for any matrix $A$, $X\mapsto AXA^*$ preserves operator order. Since partial trace preserves positive semidefiniteness, then $\Gamma$ is completely positive.
\end{proof}

We now prove that $H(\Phi(\rho))$ is large for all inputs $\rho$.
\begin{prop}[large $\Hmin(\Phi)$]\label{prop:1}
Suppose the net $\mathcal N_\eta$ in \eqref{eqn:net} is an $\varepsilon$-net with $\varepsilon<1$.
For any input density matrix $\rho$,
\begin{align}\label{eqn:lower1}
H(\Phi(\rho))&\ge \log k-\frac{L^2}{k(1-\varepsilon)^2}.
\end{align}
\end{prop}
\begin{proof}
The entropy bound \eqref{eqn:Hbound} from Section~\ref{subsec:prelim-ta} gives
\begin{align}\label{eqn:hlower}
H(\Phi(\rho))&\ge \log k-k\|\Phi(\rho)-I_k/k\|_\hs^2.
\end{align}
Recalling $D=\Phi(\rho)-I_k/k$ and $S_F=\{A:A=A^*,\|A\|_\hs=1,\operatorname{diag}(A)=0\}$, then \eqref{eqn:D} and \eqref{eqn:trad} give
\begin{align}
\|\Phi(\rho)-I_k/k\|_\hs&\le\frac{1}{k}\sup_{A\in S_F}\|FT_AF\|.
\end{align}
Let $s^*:=\sup_{A\in S_F}\|FT_AF\|<\infty$.
We can use the bound \eqref{eqn:lbound} $\|FT_{A_q}F\|\le L$, for $A_q$ in the $\varepsilon$-net $\mathcal N_\eta$ \eqref{eqn:net}, to write
\begin{align*}
\|FT_AF\|&\le \|FT_{A_q}F\| + \|FT_{A-A_q}F\|\\
&\le L+\|A-A_q\|_\hs s^*\le L+\varepsilon s^*.\numberthis
\end{align*}
Taking the supremum over $A\in S_F$ gives $s^*\le\frac{L}{1-\varepsilon}$.
Thus \eqref{eqn:hlower} gives \eqref{eqn:lower1} as desired.
\end{proof}

We next check that the maximally entangled Bell state gives small entropy $H((\Phi\otimes\Phi)(\rho_{12}))$, as it also did for the channels in e.g. \cite{hastings2009superadditivity,collins2018haagerup}. 
The main input will be the trace bound \eqref{eqn:tauF2} for $\tau_d(F^2)$.

\begin{prop}[small $\Hmin(\Phi\otimes\Phi)$]\label{prop:2}
For any $s\in\N$, let $|\Omega_s\rangle:=s^{-1/2}\sum_{j=1}^s|j\rangle\otimes|j\rangle$ denote the Bell state on the tensor product of $s$-dimensional space.
Let $\sigma_{12}:=(\Phi\otimes\Phi)(|\Omega_d\rangle\langle\Omega_d|)$ be the output state under the 2-use channel $\Phi\otimes\Phi$. Then
\begin{align}\label{eqn:lower}
\langle\Omega_k|\sigma_{12}|\Omega_k\rangle&\ge\frac{\tau_d(F^2)^2}{k}.
\end{align}
Additionally, if $\tau_d(F^2)^2\ge c_1$ and $1/k\le c_1\le1$, then
\begin{align}\label{eqn:h12}
H(\sigma_{12})&\le 2\log k-\frac{c_1\log k-1}{k}.
\end{align}
\end{prop}

\begin{proof}
First, recall from the proof of Lemma~\ref{lem:channel} that we write $\Phi=\Lambda+\Gamma$, with $[\Lambda(\rho)]_{ij}=\frac1k\Tr(U_iF\rho FU_j^*)$ the first term in \eqref{eqn:phi}. Both $\Lambda$ and $\Gamma$ were checked to be completely positive in the proof, so all four individual terms in the decomposition $\Phi\otimes\Phi=\Lambda\otimes\Lambda+\Lambda\otimes\Gamma+\Gamma\otimes\Lambda+\Gamma\otimes\Gamma$ are completely positive as well \cite[\S2]{watrous2018book}. Thus they map $|\Omega_d\rangle\langle\Omega_d|$ to another positive semidefinite operator, and so for the lower bound \eqref{eqn:lower} it suffices to consider only $(\Lambda\otimes\Lambda)(|\Omega_d\rangle\langle\Omega_d|)$.

Write $|\Omega_d\rangle\langle\Omega_d|=\frac1d\sum_{a,b=1}^d(|a\rangle\langle b|)\otimes(|a\rangle\langle b|)$, and note that $\langle\Omega_k|X\otimes Y|\Omega_k\rangle=\frac1k\sum_{i,j=1}^k\langle i|X|j\rangle\langle i|Y|j\rangle$. Using that $F$ and all $U_j$ are real, compute
\begin{align*}
\langle\Omega_k|(\Lambda\otimes\Lambda)(|\Omega_d\rangle\langle\Omega_d|)|\Omega_k\rangle&=\frac{1}{dk}\sum_{i,j=1}^k\sum_{a,b=1}^d[\Lambda(|a\rangle\langle b|)]_{ij}^2\\
&=\frac{1}{dk^3}\sum_{i,j=1}^k\|FU_j^*U_iF\|_\hs^2\\
&=\frac{1}{dk^3}\sum_{i,j=1}^k\Tr(U_iF^2U_i^*U_jF^2U_j^*)
=\frac{1}{dk^3}\Tr(B^2),\numberthis
\end{align*}
for $B:=\sum_{i=1}^kU_iF^2U_i^*$, which is positive semidefinite. Thus Cauchy--Schwarz $\Tr(B^2)\ge\frac1d(\Tr B)^2=\frac{k^2}{d}\Tr(F^2)^2$ implies
\begin{align}\label{eqn:lower2}
\langle\Omega_k|(\Lambda\otimes\Lambda)(|\Omega_d\rangle\langle\Omega_d|)|\Omega_k\rangle&\ge\frac{1}{k}\tau_d(F^2)^2,
\end{align}
which implies \eqref{eqn:lower}.

To obtain the entropy bound \eqref{eqn:h12}, from \eqref{eqn:lower}, we see that $\sigma_{12}=(\Phi\otimes\Phi)(|\Omega_d\rangle\langle\Omega_d|)$ has a largest eigenvalue $\lambda$ such that $\lambda\ge c_1/k$.
Once we have the largest eigenvalue $\lambda$, the entropy is maximized by the uniform distribution on the remaining eigenvalues, which have total mass $1-\lambda$. More precisely, if the other eigenvalues of $\sigma_{12}$ are $\lambda_2,\ldots,\lambda_{k^2}$, then by concavity of $f(x)=-x\log x$, Jensen's inequality gives
\begin{align}\label{eqn:jensen-unform}
\frac{1}{k^2-1}\sum_{j=2}^{k^2}f(\lambda_j)&\le f\Bigg(\frac{1}{k^2-1}\sum_{j=2}^{k^2}\lambda_j\Bigg)=f\left(\frac{1-\lambda}{k^2-1}\right),
\end{align}
see also \cite[\S5]{fukuda2010comments}. Thus
\begin{align}
H(\sigma_{12})&\le -\lambda\log\lambda-(1-\lambda)\log(1-\lambda)+(1-\lambda)\log(k^2-1)=:g(\lambda).
\end{align}
The right-hand side has derivative $g'(t)=\log\left(\frac{1-t}{t(k^2-1)}\right)$, which is $\le0$ for $t\ge1/k^2$. Since $\lambda\ge c_1/k\ge1/k^2$, letting $a:=c_1/k$ for notational convenience, we obtain
\begin{align*}
H(\sigma_{12})\le g(\lambda)\le g(c_1/k)&\le -a\log a+a+(1-a)2\log k\\
&\le2\log k-\frac{c_1}{k}\log k+\frac{1}{k},\numberthis\label{eqn:hg}
\end{align*}
where we used $(1-a)\log(1-a)+a\ge0$ for $0<a<1$ and $c_1-c_1\log c_1\le1$ for $0<c_1\le1$ (which are the same inequality).
\end{proof}

\subsection{Parameter choices and completion of the proof of Theorem~\ref{thm:main}}\label{subsec:parameters}

We now choose parameters which produce the explicit numerical values in the statement of Theorem~\ref{thm:main}, and use this to finish the proof of Theorem~\ref{thm:main}.

Combining Propositions~\ref{prop:1} and \ref{prop:2} gives
\begin{align}\label{eqn:phi-diff}
2\Hmin(\Phi)-\Hmin(\Phi\otimes\Phi)&\ge \frac1k\left[-\frac{2L^2}{(1-\varepsilon)^2}+c_1\log k-1\right],
\end{align}
where
\begin{itemize}
\item $L>3$ is to be chosen;
\item $c_1\le\tau_d(F^2)^2$ can be taken as $c_1=\left[1-2(2\eta+1)^r\left(\frac{3}{L}\right)^{2m}\right]^2$ by \eqref{eqn:tauF2}; for Proposition~\ref{prop:2} we also require $c_1\ge 1/k$;
\item $\varepsilon$ is the ball radius $\varepsilon=\sqrt{r}/\eta<1$ in the net $\mathcal N_\eta$, for $r=k(k-1)$.
\end{itemize}
The available parameters $L,k,\eta,m$ need to be chosen so that the right side of \eqref{eqn:phi-diff} is positive. 
Note that once we choose $k$ and $m$, then $r=k(k-1)$, $n=16m+8m\log_2k=272m$, $M=2^n$, and $d=|\SL_2(\Z/M\Z)|=3\cdot 2^{3n-2}$ are all determined.
We start by taking $L$ as follows, and asserting goals for $\varepsilon$ and $c_1$:
\begin{align}\label{eqn:param}
L=3.125=\frac{25}{8},\quad \varepsilon\le\frac{1}{64},\quad c_1\ge\frac{99}{100}.
\end{align}
We will choose parameters $\eta$ and $m$ last, so we can always ensure the goals for $\varepsilon$ and $c_1$ are met by taking $\eta$ and $m$ sufficiently large.
We now choose $k$ large enough so the gap in \eqref{eqn:phi-diff} is positive. With the parameters in \eqref{eqn:param}, taking $k=2^{32}$ for convenience gives
\begin{align}
c_1\log k-1-\frac{2L^2}{(1-\varepsilon)^2}&> \frac{3}{4},
\end{align}
and so the entropy gap \eqref{eqn:phi-diff} is $>\frac{3}{4k}=\frac{3}{17179869184}$.

We have $r=k(k-1)$, and we will choose $\eta:=64k\ge64\sqrt{k(k-1)}$, so $\varepsilon=\sqrt{r}/\eta\le 1/64$.
Then we will need to choose $m$ large enough so that $c_1\ge\frac{99}{100}$. Since
\begin{align*}
c_1=\left[1-2(128k+1)^{r}\left(\frac{24}{25}\right)^{2m}\right]^2,
\end{align*}
using $k=2^{32}$ it suffices to take
\begin{align}
m\ge\frac{r\log(2^{39}+1)+\log\Big(\frac{2}{1-\sqrt{\frac{99}{100}}}\Big)}{2\log(25/24)}
\end{align}
and we see taking e.g. $m=350r$ works.
From \eqref{eqn:d} and \eqref{eqn:M}, we then have $n=272m=1756130035408268427264000$ and
\begin{align}
d=3\cdot 2^{3n-2}=3\cdot 2^{5268390106224805281791998}.
\end{align}
This completes the proof of Theorem~\ref{thm:main}. \qed

\appendix
\section{\texorpdfstring{Minimum output R\'enyi-$p$ entropy for $p\ge1$}{Minimum output Rényi-p entropy for p≥1}}\label{sec:p}

In this section, we prove Corollary~\ref{cor:p} on additivity violation for $\Phi$ for the minimum output R\'enyi-$p$ entropy, any $1\le p\le\infty$. 
The proof will follow from two ingredients already proved in the proof of Theorem~\ref{thm:main}, namely that for any density matrix $\rho$ (from the proof of Proposition~\ref{prop:1}),
\begin{align}\label{eqn:2bound}
\|\Phi(\rho)-I_k/k\|_\hs\le\frac{L}{k(1-\varepsilon)}\le\frac{C_1}{k},\quad\text{for }C_1:=\frac{200}{63},
\end{align}
and (from the proof of Proposition~\ref{prop:2}), the output $\sigma_{12}=(\Phi\otimes\Phi)(|\Omega_d\rangle\langle\Omega_d|)$ has largest eigenvalue
\begin{align}\label{eqn:lambdabound}
\lambda\ge\frac{c}{k},\quad\text{for }c:=\frac{99}{100}.
\end{align}

\begin{proof}[Proof of Corollary~\ref{cor:p}]
For notational convenience, we will define
\begin{align}
A:=1+\frac{C_1^2}{k}, \quad B:=1+C_1,\quad a:=\frac{c}{k}.
\end{align}
We start by proving $H_p(\Phi(\rho))$ is large for all $\rho$. 
From \eqref{eqn:2bound} and since $\Tr(\Phi(\rho)-I_k/k)=0$, we have
\begin{align}\label{eqn:p-trace}
\|\Phi(\rho)\|\le \frac{B}{k},\quad\text{and}\quad\Tr(\Phi(\rho)^2)=\frac1k+\|\Phi(\rho)-I_k/k\|_\hs^2\le\frac{A}{k}.
\end{align}
We split into two cases, $1\le p\le2$, and $2\le p\le\infty$.
The second equality of \eqref{eqn:p-trace} gives a bound on the R\'enyi-2 entropy, so for $1\le p\le2$, we can use monotonicity of the R\'enyi-$p$ entropies in $p$ to write
\begin{align}\label{eqn:hp-1}
H_p(\Phi(\rho))\ge H_2(\Phi(\rho))\ge \log k-\log A.
\end{align}
For $p\ge2$, we have
\begin{align}
\Tr(\Phi(\rho)^p)&\le \|\Phi(\rho)\|^{p-2}\Tr(\Phi(\rho)^2)\le k^{1-p}B^{p-2}A,
\end{align}
which gives the bound
\begin{align}\label{eqn:Hp-lower-1}
H_p(\Phi(\rho))&\ge \log k-\frac{1}{p-1}\log A-\frac{p-2}{p-1}\log B.
\end{align}

Next, for $p>1$, we show $H_p((\Phi\otimes\Phi)(|\Omega_d\rangle\langle\Omega_d|))$ is still small. Given the largest eigenvalue $\lambda$ of $\sigma_{12}=(\Phi\otimes\Phi)(|\Omega_d\rangle\langle\Omega_d|)$, let the other eigenvalues of $\sigma_{12}$ be $\lambda_2,\ldots,\lambda_{k^2}$.
From a similar argument as in \eqref{eqn:jensen-unform}, using that $x\mapsto x^p$ is convex for $p\ge1$, we see that given the largest eigenvalue $\lambda$ of $\sigma_{12}$, the R\'enyi-$p$ entropy is maximized by taking the uniform distribution on the remaining eigenvalues, which have total mass $1-\lambda$.
In this case the entropy is given by $\frac{1}{1-p}\log F(\lambda)$ for $F(t):=t^p+(k^2-1)^{1-p}(1-t)^p$. By \eqref{eqn:lambdabound}, we have $\lambda\ge c/k=a$. We can check that $F'(t)\ge0$ for $t\ge1/k^2$; since $\lambda\ge a\ge c/k>1/k^2$, we see $F(\lambda)\ge F(a)$.
Since $1/(1-p)<0$, then letting
\begin{align}
\xi=\left(a,\frac{1-a}{k^2-1},\ldots\frac{1-a}{k^2-1}\right),\quad\text{for }a=\frac{c}{k},
\end{align}
we obtain
\begin{align}
H_p(\sigma_{12})\le H_p(\xi).
\end{align}

For $1\le p\le2$, it is again enough to use monotonicity of R\'enyi-$p$ entropy in $p$, which implies 
\begin{align}\label{eqn:H1-xibound}
H_p(\xi)\le H_1(\xi)\le 2\log k-\frac{c\log k-1}{k},
\end{align} 
using the estimates in \eqref{eqn:hg} in the proof of Proposition~\ref{prop:2}. Combining this with \eqref{eqn:hp-1} gives for $1\le p\le2$ and $k=2^{32}$,
\begin{align*}
2\Hpmin(\Phi)-\Hpmin(\Phi\otimes\Phi)&\ge \frac{c\log k-1}{k}-2\log A\\
&\ge \frac{c\log k-1}{k}-\frac{2C_1^2}{k}>\frac{3}{17\,179\,869\,184}.\numberthis
\end{align*}
This finishes the proof for the case $1\le p\le2$.

For $p\ge2$, we will do a more careful entropy estimate. Write
\begin{align}
H_p(\xi)&=\frac1{1-p}\log(a^p+(k^2-1)^{1-p}(1-a)^p)=2\log k-\frac{G(p)}{p-1},
\end{align}
for $G(p):=(p-1)2\log k+\log(a^p+(k^2-1)^{1-p}(1-a)^p)$. Then $G$ is convex, e.g. by direct differentiation of the logarithm of a sum of exponential terms, and considering $p\ge2$ gives $G(p)\ge G(2)+(p-2)G'(2)$.
Since $G(2)=2\log k-H_2(\xi)\ge 2\log k-H_1(\xi)$, we then have
\begin{align*}
H_p(\xi)&\le 2\log k-\frac{G(2)+(p-2)G'(2)}{p-1}\\
&\le2\log k-\frac{1}{p-1}\left(\frac{c\log k-1}{k}\right)-\frac{p-2}{p-1}G'(2),\numberthis
\end{align*}
using \eqref{eqn:H1-xibound} to bound $H_1(\xi)$.

We want to show $G'(2)$ is sufficiently large. Recalling $a=c/k$,
\begin{align*}
G'(2)&=2\log k+\frac{a^2\log a+(k^2-1)^{-1}(1-a)^2\log\left(\frac{1-a}{k^2-1}\right)}{a^2+(k^2-1)^{-1}(1-a)^2}\\
&=\frac{c^2\log(ck)+\frac{(k-c)^2}{k^2-1}\log\left(\frac{k(k-c)}{k^2-1}\right)}{c^2+\frac{(k-c)^2}{k^2-1}}.\numberthis\label{eqn:Gprime2}
\end{align*}
We can check that for $c=99/100$ and $k\ge2$,
\begin{align}
\frac{(k-c)^2}{k^2-1}<1,\quad\text{and}\quad \log\left(\frac{k(k-c)}{k^2-1}\right)\ge\log\left(1-\frac{1}{k}\right)>-\frac2k.
\end{align}
The first inequality implies the denominator of \eqref{eqn:Gprime2} is $<2$, so for $k=2^{32}$ we obtain
\begin{align}
G'(2)>\frac{c^2}{2}\log(ck)-\frac{2}{k}>10,
\end{align}
and
\begin{align}
H_p(\sigma_{12})&\le H_p(\xi)\le2\log k-\frac{1}{p-1}\left(\frac{c\log k-1}{k}\right) -10\cdot\frac{p-2}{p-1}.
\end{align}
Finally, combining with \eqref{eqn:Hp-lower-1}, we obtain for $p\ge2$,
\begin{align*}
2\Hpmin(\Phi)-\Hpmin(\Phi\otimes\Phi)&\ge\frac{1}{p-1}\left(\frac{c\log k-1}{k}\right) -\frac{2}{p-1}\log A+\frac{p-2}{p-1}(10-2\log B)\\
&> \frac{1}{p-1}\frac{3}{17\,179\,869\,184}+\frac{p-2}{p-1}(10-2\log B).\numberthis
\end{align*}
Since $B=\frac{263}{63}$, then $10-2\log B>7$, and $\frac{\alpha}{p-1}+\frac{7(p-2)}{p-1}\ge \alpha$ for any $\alpha\le7$ and $p\ge2$, which gives an entropy gap greater than $\frac{3}{17\,179\,869\,184}$, for $2\le p\le\infty$ as well.
\end{proof}

\section{Extension to large entropy gap and \texorpdfstring{$p\ge p_0>0$}{p≥p0>0}}\label{sec:ext}

In this section, we prove Theorem~\ref{thm:ext} on an explicit construction with large entropy gaps for the minimum output R\'enyi-$p$ entropy over all $p\ge p_0>0$. The construction is similar to that of Theorem~\ref{thm:main}, but replaces the free group $F_k$ with powers $[F_k^\gamma]^\ell$, which are used to increase and amplify the gap as in \cite{wang2026unbounded,zhen2026almost}. Haagerup's inequality Theorem~\ref{thm:haagerup} is replaced with the product form of Haagerup's inequality \cite[Proposition 3.2]{collins2022additivity}, and the damping factor $F$ is modified to suppress some additional terms to handle large $p$.

\subsection{Free powers}
Instead of the free group $F_k$ on $k$ generators, we consider $[F_k^\gamma]^\ell=\prod_{\beta=1}^\ell\prod_{\alpha=1}^\gamma F_k^{(\alpha,\beta)}$, and denote the generators of each $F_k^{(\alpha,\beta)}$ by $\{g_{1,\alpha,\beta},\ldots,g_{k,\alpha,\beta}\}$. The power $\gamma$ will be used to increase the entropy gap $O(1/k)$ to $\Theta(1)$ for $p\ge p_0$ as in \cite{zhen2026almost}, and then the power $\ell$ will be used to amplify this gap similar to \cite{wang2026unbounded}.

As a combination of \cite{zhen2026almost} and \cite{wang2026unbounded}, consider tuples $\s=(s_1,\ldots,s_\ell)$, where $s_\beta=(j_\beta,\alpha_\beta,\xi_\beta)$ for $j_\beta\in[k]$ choosing the generator, $\alpha_\beta\in[\gamma]$ choosing the copy of $F_k$ in $F_k^\gamma$, and $\xi_\beta\in\{\pm1\}$ an additional parameter which will indicate whether we take the generator $g_{j_\beta,\alpha_\beta,\beta}$ or its inverse $g_{j_\beta,\alpha_\beta,\beta}^{-1}$. These extra signs $\xi_\beta\in\{\pm1\}$ and intermediate power $\gamma$ will be used to be able to handle small $p_0$ similarly as in \cite{zhen2026almost}.
Note for each $\beta\in[\ell]$, there are $K:=2\gamma k$ possible values of $s_\beta$, so in total $N:=K^\ell=(2\gamma k)^\ell$ tuples $\s$.

For each $\s$, define the unitary $u_\s$ as the left-regular representation corresponding to $g_\s:=\prod_{\beta=1}^\ell g_{j_\beta,\alpha_\beta,\beta}^{\xi_\beta}$,
\begin{align}
u_\s&=\lambda_{[F_k^\gamma]^\ell}\left(\prod_{\beta=1}^\ell g_{j_\beta,\alpha_\beta,\beta}^{\xi_\beta}\right)
=\prod_{\beta=1}^\ell\lambda_{[F_k^\gamma]^\ell}(g_{j_\beta,\alpha_\beta,\beta}^{\xi_\beta}).
\end{align}
For an $N\times N$ zero-diagonal hermitian matrix $A$, define the function $f_A$ via 
\begin{align}
f_A(g)=\sum_{\substack{g_\s^{-1}g_\t=g,\\\s\ne\t}}A_{\s\t},
\end{align}
with $f_A(g)=0$ if the sum is empty. Note the sum in the definition of $f_A$, which ensures it is well-defined, is introduced due to $\ell\ge2$ and the signs $\xi_\beta\in\{\pm1\}$.
We can define $T_{A,\mathrm{free}}$ as in Section~\ref{sec:prelim}, as well as a new operator $L_{v,\free}$ for $v=(v_\s)\in\C^N$,
\begin{align}\label{eqn:op-ext}
T_{A,\mathrm{free}}:=\lambda_{[F_k^\gamma]^\ell}(f_A)=\sum_{\s\ne \t}A_{st}u_\s^*u_\t,\qquad L_{v,\mathrm{free}}:=\lambda_{[F_k^\gamma]^\ell}(h_v)=\sum_\s v_\s u_\s,
\end{align}
for $h_v(g_\s):=v_\s$ and zero otherwise.
The new operator $L_{v,\mathrm{free}}$ will correspond to a finite-dimensional operator $L_v$, which will be used in the new damping factor $F$ to control large eigenvalues of the channel output at large $p$.

Since we work with the product $[F_k^\gamma]^\ell$, we will need the product version of Haagerup's inequality from \cite[Proposition 3.2]{collins2022additivity}. 

\begin{thm}[product version of Haagerup's inequality {\cite[Proposition 3.2]{collins2022additivity}}]\label{thm:haagerup-product}
Consider $F_k^q$, and for $\mathbf n=(n_1,\ldots,n_q)\in\N_0^q$, let $E_{\mathbf n}:=\{(g_1,\ldots,g_q)\in F_k^q:|g_i|=n_i\;\forall i\}$.
Then for $f\in\ell^2(F_k^q)$ supported in $E_{\mathbf n}$ and $\lambda(f):=\sum_{g\in F_k^q}f(g)\lambda(g)$, there is the operator norm bound
\begin{align}
\|\lambda(f)\|\le\left[\prod_{i=1}^q(n_i+1)\right]\|f\|_2.
\end{align}
\end{thm}

Similar to the methods of \cite[Lemma 2.1]{wang2026unbounded} and \cite[Proposition 5.1]{zhen2026almost}, we can apply Theorem~\ref{thm:haagerup-product} to obtain bounds on the operators in \eqref{eqn:op-ext}.
\begin{lem}[consequence of {\cite[Proposition 3.2]{collins2022additivity}}]\label{lem:ph}
Consider $[F_k^\gamma]^\ell$, and let $K=2\gamma k$ and $N=K^\ell$.
For $N\times N$ zero-diagonal hermitian $A$,
\begin{align}
\|T_{A,\mathrm{free}}\|\le J\|A\|_\hs,\quad\text{for } J=\sqrt{(K+16\gamma^2)^\ell-K^{\ell}},
\end{align}
and for $v\in\C^N$,
\begin{align}\label{eqn:Lv-haagerup}
\|L_{v,\mathrm{free}}\|\le(2\sqrt{\gamma})^\ell\|v\|_2.
\end{align}
\end{lem}
\begin{proof}
We start with $L_{v,\free}$ since it is simpler. 
Recall $\s=(s_1,\ldots,s_\ell)$ with $s_\beta=(j_\beta,\alpha_\beta,\xi_\beta)$ for $j_\beta\in[k]$ and $\alpha_\beta\in[\gamma]$.
In order to apply Theorem~\ref{thm:haagerup-product} to $L_{v,\free}$, we will split it into operators $L_{v,\free}^\alpha$, each of which is supported on a single subspace $E_\mathbf{n}\subset[F_k^\gamma]^\ell\cong F_k^{\gamma\ell}$, for $\mathbf n=(n_{\alpha',\beta})_{\alpha',\beta}\in\{0,1\}^{\gamma\ell}$.
Let $\boldsymbol\alpha(\s):=(\alpha_1,\cdots,\alpha_\ell)\in[\gamma]^\ell$, and partition $L_{v,\free}$ as
\begin{align}
L_{v,\free}=\sum_{\alpha\in[\gamma]^\ell}\sum_{\s:\boldsymbol\alpha(\s)=\alpha}v_\s u_\s=\sum_{\alpha\in[\gamma]^\ell}L_{v,\free}^{(\alpha)}.
\end{align}
Then $L_{v,\free}^{(\alpha)}$ is supported in $E_{\mathbf n}$ for $\mathbf n=(n_{\alpha',\beta})_{\alpha',\beta}$ defined as $n_{\alpha',\beta}=1$ if $\alpha'=\alpha_\beta$, and 0 otherwise.
We can write $L_{v,\free}^{(\alpha)}=\lambda_{[F_k^\gamma]^\ell}(h^{(\alpha)})$, where $h^{(\alpha)}(g_\s):=v_\s$ for $\s$ such that $\boldsymbol\alpha(\s)=\alpha$, and zero otherwise. 
Applying Haagerup's product inequality Theorem~\ref{thm:haagerup-product} with $q=\gamma\ell$ to $L_{v,\free}^{(\alpha)}$ then gives
\begin{align}
\|L_{v,\free}^{(\alpha)}\|= \|\lambda(h^{(\alpha)})\|\le\left(\prod_{\beta=1}^\ell2\right)\|h^{(\alpha)}\|_2=2^\ell\|h^{(\alpha)}\|_2=2^\ell\|v^{(\alpha)}\|_2,
\end{align}
for $(v^{(\alpha)})_\s:=v_\s$ if $\boldsymbol\alpha(\s)=\alpha$ and $0$ otherwise.
Then summing over $\alpha\in[\gamma]^\ell$ with the triangle inequality and Cauchy--Schwarz gives
\begin{align}
\|L_{v,\free}\|&\le \sum_{\alpha\in[\gamma]^\ell}2^\ell\|v^{(\alpha)}\|_2\le 2^\ell\left(\sum_{\alpha\in[\gamma]^\ell}\|v^{(\alpha)}\|_2^2\right)^{1/2}\gamma^{\ell/2}=(2\sqrt{\gamma})^{\ell}\|v\|_2,
\end{align}
which is \eqref{eqn:Lv-haagerup}.

For $T_{A,\free}$, the support of $f_A$ is more complicated. The method below is similar to \cite[Proposition 5.1]{zhen2026almost}, but with the additional $\ell$ factors. Consider a fixed $\beta\in[\ell]$. For a pair $(s_\beta,t_\beta)$, there are three possibilities:
\begin{enumerate}
\item $E$: $s_\beta=t_\beta$
\item $S_\alpha$: $s_\beta\ne t_\beta$ but they are in the same factor $F_k$ in $F_k^\gamma$, indexed by $\alpha\in[\gamma]$.
\item $D_{\{\alpha,\alpha'\}}$: $s_\beta$ and $t_\beta$ lie in distinct copies of $F_k$ in $F_k^\gamma$, corresponding to indices $\alpha,\alpha'\in[\gamma]$.
\end{enumerate}
The possible types for the pair $(s_\beta,t_\beta)$ are then $\omega\in\{E\}\cup\{S_\alpha\}_{\alpha=1}^\gamma\cup\{D_{\alpha,\alpha'}\}$.
A full pair $(\s,\t)$ can thus be described by a pattern $\omega=(\omega_1,\ldots,\omega_\ell)$ in $(\{E\}\cup\{S_\alpha\}_{\alpha=1}^\gamma\cup\{D_{\alpha,\alpha'}\}_{\alpha<\alpha'})^\ell$. Split $f_A$ as
\begin{align}
f_A(g)=\sum_{\omega\ne(E,\ldots,E)}\sum_{\substack{g_\s^{-1}g_\t=g\\(\s,\t)\text{ type }\omega}}A_{\s\t}=:\sum_{\omega\ne(E,\ldots,E)}f_{A,\omega}(g).
\end{align}
The function $f_{A,\omega}$ is supported on a single $E_{\mathbf n}\subset\prod_{\alpha=1}^\gamma\prod_{\beta=1}^\ell F_k^{(\alpha,\beta)}\cong F_k^{\gamma\ell}$ for $\mathbf n=(n_{\alpha,\beta})_{\alpha,\beta}\in\{0,1,2\}^{\gamma\ell}$ depending on $\omega$.
The operator $T_{A,\free}$ splits as $T_{A,\free}=\sum_{\omega\ne(E,\ldots,E)}\lambda(f_{A,\omega})$, and Theorem~\ref{thm:haagerup-product} applied with $q=\gamma\ell$ gives
\begin{align}\label{eqn:lf-1}
\|\lambda(f_{A,\omega})\|&\le 3^{s(\omega)}4^{d(\omega)}\|f_{A,\omega}\|_2,
\end{align}
where we let
\begin{enumerate}
\item $e(\omega):=\#\{\omega_\beta=E\}$ (for later use);
\item $s(\omega):=\#\{\omega_\beta\in \{S_\alpha\}_\alpha\}$ 
(contributing a factor $2+1=3$);
\item $d(\omega):=\#\{\omega_\beta\in\{D_{\alpha,\alpha'}\}_{\alpha,\alpha'}\}$ (contributing a factor $(1+1)(1+1)=4$).
\end{enumerate}
We next relate $\|f_{A,\omega}\|_2$ to $\|A_\omega\|_\hs$, for $A_\omega$ defined as $(A_\omega)_{\s\t}:=A_{\s\t}$ if $(\s,\t)$ has type $\omega$, and 0 otherwise. For this, we need to count how many $(\s,\t)$ there are of type $\omega$ with $g_\s^{-1}g_\t=g$. For a fixed $\beta$ and $s_\beta=(j_\beta,\alpha_\beta,\xi_\beta)$, let $g_{s_\beta}:=g_{j_\beta,\alpha_\beta,\beta}^{\xi_\beta}$, which we will temporarily also denote by $s_\beta$. Then we can count:
\begin{enumerate}
\item If $\omega_\beta=E$, then there are $K=2k\gamma$ choices for the value of $s_\beta=t_\beta$.
\item If $\omega_\beta=S_\alpha$, then the map $(s_\beta,t_\beta)\mapsto s_\beta^{-1}t_\beta$ is injective, so the multiplicity is 1.
\item If $\omega_\beta= D_{\{\alpha,\alpha'\}}$, then $s_\beta$ and $t_\beta$ commute, so both $(s_\beta,t_\beta)$ and $(t_\beta^{-1},s_\beta^{-1})$ give the same $g$, and the multiplicity is 2.
\end{enumerate}
Applying Cauchy--Schwarz with these multiplicities gives
\begin{align}
|f_{A,\omega}(g)|^2=\Bigg|\sum_{\substack{g_\s^{-1}g_\t=g\\(\s,\t)\text{ type }\omega}}A_{\s\t}\Bigg|^2&\le K^{e(\omega)}2^{d(\omega)}\sum_{\substack{g_\s^{-1}g_\t=g\\(\s,\t)\text{ type }\omega}}|A_{\s\t}|^2. 
\end{align}
Thus summing over $g\in [F_k^\gamma]^\ell$ and using \eqref{eqn:lf-1}, we have
\begin{align}
\|\lambda(f_{A,\omega})\|^2&\le K^{e(\omega)}9^{s(\omega)}32^{d(\omega)}\|A_\omega\|_\hs^2.
\end{align}
This gives
\begin{align*}
\|T_{A,\free}\|&\le\sum_{\omega\ne(E,\ldots,E)}\|\lambda(f_{A,\omega})\|\\
&\le \sum_{\omega\ne(E,\ldots,E)}K^{e(\omega)/2}3^{s(\omega)}32^{d(\omega)/2}\|A_\omega\|_\hs\\
&\le \left(\sum_{\omega\ne(E,\ldots,E)}K^{e(\omega)}9^{s(\omega)}32^{d(\omega)}\right)^{1/2}\|A\|_\hs.\numberthis
\end{align*}
We just need to evaluate the sum over $\omega$. For each $\beta$, we count how many choices there are for $\omega_\beta$ depending on its type. For $E$, there is only 1 choice of $E$. For $S_\alpha$, there are $\gamma$ choices of $\alpha$. For $D_{\{\alpha,\alpha'\}}$, there are $\tbinom{\gamma}{2}$ choices of $\{\alpha,\alpha'\}$. Thus using the multinomial theorem, we obtain, summing over all $\omega$,
\begin{align}
\sum_{\omega}K^{e(\omega)}9^{s(\omega)}32^{d(\omega)}&=\sum_{e+s+d=\ell}K^e9^s32^d\gamma^s\binom{\gamma}{2}^d\binom{\ell}{e,s,d}=\left(K+9\gamma+32\binom{\gamma}{2}\right)^\ell.
\end{align}
Using $9\gamma+32\tbinom\gamma2=16\gamma^2-7\gamma\le 16\gamma^2$, and subtracting $K^\ell$ to remove the $\omega=(E,\ldots,E)$ term, finishes the proof of Lemma~\ref{lem:ph}.
\end{proof}

\subsection{Explicit construction} \label{subsec:ext-construction}
Next, we reduce $[F_k^\gamma]^\ell$ to finite dimensions.
As in Section~\ref{subsec:construction}, we use the Sanov representation \eqref{eqn:sanov} of $F_2$ in $\SL_2(\Z)$. For each copy $(\alpha,\beta)$ of $F_k$, we use generators $\{a_{j,\alpha,\beta}\}_{j=1}^k$ which are defined as in \eqref{eqn:aj}, and which commute with generators from other copies.
For a parameter $m$ to be chosen later in Appendix~\ref{subsec:param-ext}, we consider the matrix entries modulo a large integer $M=2^n$, also to be chosen later, so that nontrivial reduced words in the $a_{j,\alpha,\beta}$ of length $\le 4 m$ in each $F_k^{(\alpha,\beta)}$ remain nontrivial modulo $M$. From \eqref{eqn:M} of Section~\ref{subsec:construction}, we see we can take 
\begin{align}\label{eqn:M-ext}
M:=(16k^2)^{4m}=2^n,\quad n=16m+8m\log_2k,
\end{align}
and the resulting group is
\begin{align}\label{eqn:G-ext}
G:=\SL_2(\Z/M\Z)^{\gamma\ell},\quad\text{which has size}\quad D:=|\SL_2(\Z/M\Z)|^{\gamma\ell}=(3\cdot 2^{3n-2})^{\gamma\ell}.
\end{align}
The unitaries $W_\s$ for the channel are defined via the left-regular representation $\lambda_G$ of $G$, as
\begin{align}
W_\s:=\prod_{\beta=1}^\ell\lambda_G(a_{j_\beta,\alpha_\beta,\beta}^{\xi_\beta}),
\end{align}
with normalized trace $\tau_D(X):=\frac1D\Tr(X)$.
As in Lemma~\ref{lem:moments}, the choice of $M$ ensures that nontrivial reduced words in the $a_{j_\beta,\alpha_\beta,\beta}^{\pm1}$ of length at most $4m$ in each $F_k^{(\alpha,\beta)}$ remain nontrivial in $G$.

The channel $\Psi$ in \eqref{eqn:phi2} is the same form as $\Phi$ in \eqref{eqn:phi}, but the damping term will be $F=(I+R)^{-1}$ with $R$ having an extra term. In addition to suppressing subspaces where $T_A$ is large, we will want to control the largest eigenvalue of $\Psi(\rho)$ when considering large $p$.\footnote{If one only wanted to consider large entropy gap at e.g. $p=1$, then it should not be necessary to define $L_v$.} To this end, define
\begin{align}
T_A:=\sum_{\s,\t}A_{\s\t}W_\s^*W_\t,\qquad L_v:=\sum_\s v_\s W_\s.
\end{align}

Consider $A$ hermitian, zero-diagonal, and normalized to $\|A\|_\hs=1$, and $v$ with $\|v\|_2=1$.
By the choice of $M$, there are then the moment matching bounds combined with Lemma~\ref{lem:ph}, giving
\begin{align}\label{eqn:moment-ext}
\begin{aligned}
\tau_D(T_A^{2m})&=\tau_{\mathrm{free}}(T_{A,\mathrm{free}}^{2m})\le \|T_{A,\mathrm{free}}^{2m}\|\le J^{2m},\\
\tau_D((L_v^*L_v)^{2m})&=\tau_{\mathrm{free}}((L_{v,\mathrm{free}}^*L_{v,\mathrm{free}})^{2m})\le \|(L_{v,\mathrm{free}}^*L_{v,\mathrm{free}})^{2m}\|\le (4\gamma)^{2\ell m}.
\end{aligned}
\end{align}

To construct the damping term $F$, we will use the net $\mathcal N_h$ from \eqref{eqn:net}, which is an $\varepsilon:=\sqrt{r}/h$-net for $r=N(N-1)$, and the similar net in $\C^N$, $\mathcal M_h:=\{z/\|z\|_2:0\ne z\in(\Z+i\Z)^N:|\Re z|,|\Im z|\le h\}$ for $L_v$. We will later take $\varepsilon\le1/8$.
Because we will want to control the largest eigenvalues of $\Psi(\rho)$ in Proposition~\ref{prop:1ext} for large $p$, we add an extra term to $R$ compared to \eqref{eqn:R}, defining, for $L>J$ and $B>(4\gamma)^\ell$ to be chosen in Appendix~\ref{subsec:param-ext},
\begin{align}\label{eqn:R-ext}
R:=\sum_{A_q\in\mathcal N_h}\left(\frac{T_{A_q}}{L}\right)^{2m}+\sum_{v\in\mathcal M_h}\left(\frac{L_{v}^*L_{v}}{B}\right)^{2m},\quad F:=(I+R)^{-1}.
\end{align}
We will see in \eqref{eqn:output-norm} how this controls the large eigenvalues of $\Psi(\rho)$ for any input $\rho$.
Analogous to \eqref{eqn:tauF2}, using the moment matching \eqref{eqn:moment-ext} and product Haagerup's inequality Lemma~\ref{lem:ph}, we have
\begin{align}\label{eqn:tauR-ext}
\tau_D(R)&\le (2h+1)^r\left(\frac{J}{L}\right)^{2m}+(2h+1)^{2N}\left(\frac{(4\gamma)^\ell}{B}\right)^{2m}.
\end{align}
Additionally, $0<F\le1$ and is real, as before, and the same functional calculus in Lemma~\ref{lem:F} gives $\|FT_{A_q}F\|\le L$.

We let $0<\eta\le1/2$ be a convenient parameter to be specified in Appendix~\ref{subsec:param-ext}.
We will later choose $m$ (depending on $h,J,L,\gamma,\ell,B,N,\eta$) in Appendix~\ref{subsec:param-ext} so that \eqref{eqn:tauR-ext} ensures
\begin{align}\label{eqn:tau-eta}
\tau_D(R)\le \frac{1}{64}\eta^2.
\end{align}

\subsection{Entropy bounds}

In this section, we prove large $H_p(\Psi(\rho))$ for every $\rho$, and small $H_p((\Psi\otimes\Psi)(\rho_{12}))$ for the entangled Bell state $\rho_{12}$. 

The proof of large $H_p(\Psi(\rho))$ for every $\rho$ is nearly the same as in Proposition~\ref{prop:1} and Corollary~\ref{cor:p}, with Lemma~\ref{lem:ph} replacing Haagerup's inequality Theorem~\ref{thm:haagerup}. The main difference is that due to considering large entropy gaps for large $p$ and because $L$ scales differently, it turns out we will want to use the largest eigenvalue suppression from $L_v$ to get a better bound on $\|\Psi(\rho)\|$ compared to the method leading to \eqref{eqn:p-trace}.

\begin{prop}[large $\Hpmin(\Psi)$]\label{prop:1ext}
Let 
\begin{align}\label{eqn:ab}
A_\ell:=1+\frac{L^2}{N(1-\varepsilon)^2},\quad B_\ell:=1+\frac{B}{(1-\varepsilon)^2},
\end{align}
where $L$ and $B$ are as in the definition \eqref{eqn:R-ext} of $R$, and will depend on $\ell$.
Then for any density matrix $\rho$,
\begin{align}
H_p(\Psi(\rho))&\ge \begin{cases}\log N-\log A_\ell,&0<p\le 2\\
\log N-\frac{1}{p-1}\log A_\ell-\frac{p-2}{p-1}\log B_\ell,&p\ge2
\end{cases}.
\end{align}
\end{prop}
\begin{proof}
This is analogous to \eqref{eqn:hp-1} and \eqref{eqn:Hp-lower-1}.
The proof of Proposition~\ref{prop:1} and $\|F T_{A_q}F\|\le L$ give the bound for $0<p\le2$. 
So we just need to check the bound for $p\ge2$ involving $B_\ell$. Comparing to the first bound in \eqref{eqn:p-trace}, it is sufficient to show $\|\Psi(\rho)\|\le \frac{B_\ell}{N}$.
To see this holds, first note that for $v\in\mathcal M_h$ an $\varepsilon$-net, the same argument (from the definition of $R$) as in the proof of Lemma~\ref{lem:F} gives $\|FL_v^*L_vF\|\le B$. The argument in the proof of Proposition~\ref{prop:1} for $L_vF$ then gives for any unit $v$,
\begin{align}
\|L_vF\|^2=\|FL_v^*L_vF\|\le \frac{B}{(1-\varepsilon)^2}.
\end{align}
We can then estimate for $\|x\|_2=1$ and $\bar x$ the complex conjugate of $x$,
\begin{align*}
\langle x|\Psi(\rho)|x\rangle&=\frac1N\sum_{\s,\t}\bar x_\s x_\t\Tr(W_\s F\rho FW_\t^*)+\frac1N\Tr((I-F^2)\rho)\\
&= \frac1N\Tr(L_{\bar x} F\rho FL_{\bar x}^*) +\frac{1}{N}\Tr((I-F^2)\rho)\\
&\le \frac1N\|FL_{\bar x}^*L_{\bar x}F\|+\frac1N\le \frac{B}{N(1-\varepsilon)^2}+\frac1N,\numberthis\label{eqn:output-norm}
\end{align*}
which gives the result with $B_\ell$ in \eqref{eqn:ab}.
\end{proof}

The proof for small entropy on the Bell state however requires more work than for the analogous Proposition~\ref{prop:2}. We will estimate the entropy by comparing $Z:=\Psi\otimes\Psi(\rho_{12})$ to a simpler state $Z_0$ defined below.

\begin{prop}[small $\Hpmin(\Psi\otimes\Psi)$]\label{prop:2ext}
Let $\rho_{12}=|\Omega_D\rangle\langle\Omega_D|$ be the maximally entangled Bell state, for $|\Omega_D\rangle=D^{-1/2}\sum_{j=1}^D|j\rangle\otimes|j\rangle$.
Let $\Psi_0$ be the channel defined by $[\Psi_0(\rho)]_{\s,\t}=\frac1N\Tr(W_\s\rho W_\t^*)$, which corresponds to \eqref{eqn:phi2} with $F=I$.
Let $Z_0=(\Psi_0\otimes\Psi_0)(\rho_{12})$, and let $Q_1$ be as in \eqref{eqn:Q1} below. Then for $\eta N\le1/2$ and if \eqref{eqn:tau-eta} holds,
\begin{align}
H_p((\Psi\otimes\Psi)(\rho_{12}))\le \begin{cases}
\ell\log Q_1+2\cdot 2^\ell\eta^{\min(p_0,1/2)},&p_0\le p\le 2\\
H_p(Z_0)+4\eta N,&p\ge2
\end{cases}.
\end{align}
\end{prop}

We can estimate $H_p(Z_0)$ using Lemma~\ref{lem:bell0} below, which follows from the $\ell=1$ case in \cite[Corollary 4.5]{zhen2026almost}, and makes use of the signs $\xi_\beta\in\{\pm1\}$. We will also use part of the lemma in the proof of Proposition~\ref{prop:2ext}.
\begin{lem}\label{lem:bell0} 
Consider $[F_k^\gamma]^\ell$, let $K=2\gamma k$, and let $\Psi_0$ be the channel defined by $[\Psi_0(\rho)]_{\s,\t}=\frac1N\Tr(W_\s\rho W_\t^*)$, which corresponds to \eqref{eqn:phi2} with $F=I$.
Let $Z_0:=(\Psi_0\otimes\Psi_0)(\rho_{12})$, and define 
\begin{align}\label{eqn:spectrum}
\mathfrak b_1=\left(\frac1K,\left(\frac{2}{K^2}\right)^{\times \frac{\gamma-1}{2\gamma}K^2},\left(\frac{1}{K^2}\right)^{\times \frac{K^2}{\gamma}-K}\right),
\end{align}
and
\begin{align}\label{eqn:Q1}
Q_1:=\frac{\gamma+1}{2\gamma}K^2-K+1\ge \frac{K^2}{2}.
\end{align}
Then for any $0\le p\le\infty$,
\begin{align}\label{eqn:z0-bounds}
\rank Z_0\le Q_1^\ell,\quad\text{and}\quad H_p(Z_0)\le \ell H_p(\mathfrak b_1).
\end{align}
\end{lem}
\begin{proof}
In the case $\ell=1$, \eqref{eqn:z0-bounds} follows from \cite[Corollary 4.5]{zhen2026almost}. The case $\ell\ge1$ then follows since $\ell^2(G)$, $W_\s$, $\Psi_0$ (e.g. check on product inputs), $|\Omega_D\rangle$, and $Z_0$ factor as tensor products.
For the inequality in \eqref{eqn:Q1}, using $K=2\gamma k$, note that $\frac{1}{2\gamma}K^2-K+1=2\gamma k(k-1)+1>0$.
\end{proof}

\begin{proof}[Proof of Proposition~\ref{prop:2ext}]
The first goal is to show that $Z:=(\Psi\otimes\Psi)(\rho_{12})$ does not differ that much from $Z_0=(\Psi_0\otimes\Psi_0)(\rho_{12})$, i.e. that the damping term $F$ stays similar enough to the identity, even with the new changes compared to Theorem~\ref{thm:main}.

Decompose the channel $\Psi$ into $\Lambda+\Gamma$ as in Lemma~\ref{lem:channel}, with $[\Lambda(\rho)]_{\s\t}=\frac1N\Tr(W_\s F\rho FW_\t^*)$. Letting $Z_F:=(\Lambda\otimes\Lambda)(\rho_{12})$, we will compare the trace norms $\|Z-Z_0\|_1\le \|Z-Z_F\|_1+\|Z_F-Z_0\|_1$.
Note that $\Lambda(\rho)=\Psi_0(F\rho F)$, and so
\begin{align*}
Z_F=(\Lambda\otimes\Lambda)(\rho_{12})=(\Psi_0\otimes\Psi_0)((F\otimes F)\rho_{12}(F\otimes F)).
\end{align*}
Since $\Psi_0\otimes\Psi_0$ is a channel, we have $\|Z_F-Z_0\|_1\le \|(F\otimes F)\rho_{12}(F\otimes F)-\rho_{12}\|_1$ \cite[Corollary 3.40]{watrous2018book}.
As $\Tr(\rho_{12}(F^2\otimes F^2))=\tau_D(F^4)$ since $F$ is real and self-adjoint, the gentle measurement lemma \cite{winter1999coding} gives
\begin{align}
\|Z_F-Z_0\|_1&\le \sqrt{8}\sqrt{1-\tau_D(F^4)}\le 2\sqrt{8}\sqrt{\tau_D(R)},
\end{align}
as $I-F^4=I-(I+R)^{-4}\le 4R$.
Since $Z-Z_F\ge0$ (cf. proof of Lemma~\ref{lem:channel}), we see $\|Z-Z_F\|_1=\Tr(Z-Z_F)=1-\tau_D(F^4)$, and obtain
\begin{align}\label{eqn:zdiff}
\|Z-Z_0\|_1\le (1+\sqrt{8})\sqrt{1-\tau_D(F^4)}\le 2(1+\sqrt{8})\sqrt{\tau_D(R)}<\eta,
\end{align}
using \eqref{eqn:tau-eta} for the last inequality.

Now we apply \eqref{eqn:zdiff} to prove the entropy bounds. 
\begin{itemize}[leftmargin=*]
\item For $p_0\le p\le2$, by monotonicity of R\'enyi-$p$ entropies, it suffices to prove an upper bound on $H_{p_0}((\Psi\otimes\Psi)(\rho_{12}))$. Without loss of generality we may assume $0<p_0\le1/2$, and simply replace $p_0$ by $\min(p_0,1/2)$ at the end.
For $Z_0$, we have $H_p(Z_0)\le H_0(Z_0)=\log\rank Z_0\le \ell\log Q_1$ by Lemma~\ref{lem:bell0}. To transfer this to a bound for $Z$, we show that most of the mass for $Z$ lies in the support $P_0$ of $Z_0$ (even if the rank of $Z$ becomes large). By Schatten H\"older's inequality and \eqref{eqn:zdiff}, we have for $P_0,P_0^\perp$ also denoting the orthogonal projection onto $\supp Z_0$ and $(\supp Z_0)^\perp$ respectively,
\begin{align}\label{eqn:porth}
\Tr(P_0^\perp Z)=\Tr(P_0^\perp(Z-Z_0))\le \|Z-Z_0\|_1\le\eta.
\end{align}
Let $\mathcal P(Z):=P_0ZP_0+P_0^\perp ZP_0^\perp$, which is a pinching channel.
Applying \cite[Proposition 4.6, Theorem 4.32]{watrous2018book} shows that the eigenvalues of $Z$ majorize those of $\mathcal P(Z)$.
Since $0<p_0<1$, majorization ($\lambda(\mathcal P(Z))=D\lambda(Z)$ for some doubly stochastic matrix $D$) and Jensen's inequality, followed by noting $\mathcal P(Z)=P_0ZP_0\oplus P_0^\perp ZP_0^\perp$ is block diagonal, give
\begin{align}
\Tr Z^{p_0}&\le\Tr \mathcal P(Z)^{p_0}= \Tr(P_0ZP_0)^{p_0}+\Tr (P_0^\perp ZP_0^\perp)^{p_0}.
\end{align}
For $X\ge0$ with $\rank X\le r$, H\"older's inequality also implies for $0<p_0<1$, $\Tr X^{p_0}\le r^{1-p_0}(\Tr X)^{p_0}$; thus with \eqref{eqn:porth} we obtain
\begin{align*}
\Tr Z^{p_0}&\le Q_1^{\ell(1-p_0)}+N^{2(1-p_0)}\eta^{p_0}\\
&\le Q_1^{\ell(1-p_0)}(1+2^\ell\eta^{p_0}),\numberthis
\end{align*}
using that $N^2=K^{2\ell}\le 2^\ell Q_1^\ell$ by \eqref{eqn:Q1}. Thus for $0<p_0\le1/2$,
\begin{align}
H_{p_0}(Z)&\le \ell\log Q_1+\frac{1}{1-p_0}\log(1+2^\ell\eta^{p_0})\le \ell\log Q_1+2\cdot 2^\ell\eta^{p_0}.
\end{align}

\item For $p\ge2$, recall the Schatten-$p$ norms are $\|X\|_p=\left(\sum_{j=1}^{N^2}\sigma_j(X)^p\right)^{1/p}$ for $\sigma_j(X)$ the singular values of $X$, and that $H_p(X)=\frac{p}{1-p}\log\|X\|_p$. Using \eqref{eqn:zdiff} we can estimate
\begin{align*}
\|Z\|_p&\ge\|Z_0\|_p-\|Z-Z_0\|_p\ge \|Z_0\|_p-\eta.
\end{align*}
Lemma~\ref{lem:bell0} with $p=\infty$ gives $\|Z_0\|_\infty\ge1/N$, so $\|Z_0\|_p\ge\|Z_0\|_\infty\ge 1/N$, and $\frac{\eta}{\|Z_0\|_p}\le\eta N\le\frac12$. Using $-\log(1-x)\le 2x$ then implies
\begin{align}
H_p(Z)=\frac{p}{1-p}\log\|Z\|_p&\le\frac{p}{1-p}\log(\|Z_0\|_p-\eta)\le H_p(Z_0)+4\eta N,
\end{align}
as desired.
\end{itemize}
\end{proof}

\subsection{Parameter choices and completion of the proof of Theorem~\ref{thm:ext}}\label{subsec:param-ext}

From Propositions~\ref{prop:1ext} and \ref{prop:2ext}, we obtain
\begin{align}\label{eqn:entropy-diff}
2\Hpmin(\Psi)-\Hpmin(\Psi\otimes\Psi)&\ge\begin{cases}2\log N-2\log A_\ell-\ell\log Q_1-2\cdot 2^\ell\eta^{\min(1/2,p_0)},&p_0\le p\le 2\\
2\log N-\frac{2}{p-1}\log A_\ell-\frac{2(p-2)}{p-1}\log B_\ell-H_p(Z_0)-4\eta N,&p\ge2
\end{cases}.
\end{align}
We want to choose parameters so the lower bound on the right hand side is of order $\ell$.
Recall $0<\eta\le\frac{1}{2N}$ for Proposition~\ref{prop:2ext}, and that we need choose $m$ so the bound \eqref{eqn:tau-eta} $\tau_D(R)\le\frac{1}{64}\eta^2$ holds. We do this at the very end after choosing the other parameters as follows, making no attempt to optimize values.
\begin{itemize}
\item Recall $N=(2\gamma k)^\ell$. We start by taking $\gamma=2^4=16$ and $k=2^{12}=4096$; we want to choose these fairly small to have smaller Hilbert space dimensions, but they are fairly arbitrary as long as they are large enough.
For example, this choice of $\gamma$ gives close to the leading order value of $Q_1$ in $\gamma$.
With these choices, we also have $N=2^{17\ell}$, $K=2\gamma k=2^{17}$, and
\begin{align}
J^2=(K+16\gamma^2)^\ell-K^\ell=N\left[\left(\frac{33}{32}\right)^\ell-1\right].
\end{align}

\item We next take $L$ and $B$ conveniently as follows, and assert the goal for $\varepsilon$:
\begin{align}
L=2J,\quad B=2(4\gamma)^\ell=2\cdot 64^\ell,\quad \varepsilon\le \frac18.
\end{align}
The nets $\mathcal N_h$ and $\mathcal M_h$ from the definition \eqref{eqn:R-ext} of $R$ are both $\varepsilon=\sqrt{r}/h$-nets, where $r=N(N-1)$. So we see we can take $h=8N$. 
With these choices of $L,B,\varepsilon$, the parameters $A_\ell,B_\ell$ from Proposition~\ref{prop:1ext} have bounds
\begin{align}
A_\ell &\le1+4\left[\left(\frac{33}{32}\right)^\ell-1\right]\frac{64}{49}\le 6\left(\frac{33}{32}\right)^\ell,\quad B_\ell\le1+\frac{128}{49}64^\ell<4\cdot 64^\ell.
\end{align}

\item We choose $\eta=2^{-t}$ small so that the terms $2^\ell\eta^{\min(p_0,1/2)}$ and $N\eta$ from \eqref{eqn:entropy-diff} are both $\le\frac{1}{128}=2^{-7}$. 
Since $N=2^{17\ell}$, we can choose
\begin{align}
\eta=2^{-t},\quad t=t(\ell,p_0)=\left\lceil\max\left\{\frac{\ell+7}{\min(p_0,1/2)},17\ell+7\right\}\right\rceil.
\end{align}

\item Finally, we choose $m$ so the bound $\tau_D(R)\le\frac{1}{64}\eta^2$ from \eqref{eqn:tauR-ext} holds. We can bound the net sizes $|\mathcal N_h|+|\mathcal M_h|$ by
\begin{align}
(2h+1)^r+(2h+1)^{2N}=(2^{17\ell+4}+1)^{N(N-1)}+(2^{17\ell+4}+1)^{2N}\le 2(2^{17\ell+5})^{N^2}=:S.
\end{align}
Since $h=2^{17\ell+3}$, $J/L=1/2$ and $(4\gamma)^\ell/B=1/2$, we see we can take
\begin{align}\label{eqn:mpl-ext}
m=\left\lceil\log_4\frac{64S}{\eta^2}\right\rceil
=3+t(\ell,p_0)+1+(17\ell+5)2^{34\ell-1},
\end{align}
which gives $n=16m+8m\log_2k=112m$ and input Hilbert space dimension from \eqref{eqn:G-ext},
\begin{align*}
D=(3\cdot 2^{336m-2})^{16\ell},
\end{align*}
in agreement with the hypotheses in Theorem~\ref{thm:ext}.

\end{itemize}

Now that we have chosen parameters, we estimate the entropy gap in \eqref{eqn:entropy-diff}. From \eqref{eqn:Q1} and $K=2^{17}$, we have
\begin{align}
Q_1&=\frac{\gamma+1}{2\gamma}K^2-K+1\le \frac{17}{32}K^2=17\cdot 2^{29}.
\end{align}
Starting with $p_0\le p\le2$, \eqref{eqn:entropy-diff} then gives
\begin{align*}
2\Hpmin(\Psi)-\Hpmin(\Psi\otimes\Psi)&\ge 34\ell\log 2-2\ell\log\frac{33}{32}-2\log6-\ell\log(17\cdot 2^{29})-\frac{1}{64}\\
&> \left(\frac{4}{5}\ell-6\right)\log 2,\numberthis
\end{align*}
as desired.

For $p\ge2$, the estimates \cite[Eqs.~(53),(55)]{zhen2026almost} for $H_p(\mathfrak b_1)$, multiplied by a factor of $\ell$ from Lemma~\ref{lem:bell0}, give for $K=2^{17}$,
\begin{align}
H_p(Z_0)\le\ell H_p(\mathfrak b_1)&\le 34\ell\log2-\frac{\ell}{p-1}\left[\log4-\frac34\log3+\frac14\log\frac{15}{8}+13(p-2)\log2\right].
\end{align}
Using this in \eqref{eqn:entropy-diff}, we thus obtain
\begin{align*}
2\Hpmin(\Psi)-\Hpmin(\Psi\otimes\Psi)&\ge
\begin{multlined}[t]
34\ell\log2-\frac{2}{p-1}\left[\ell\log\frac{33}{32}+\log6\right]\\
-\frac{2(p-2)}{p-1}\left[\log4+\ell\log64\right]-\frac{1}{32}-H_p(Z_0)
\end{multlined}\\
&\ge \begin{multlined}[t]\frac{1}{p-1}\left[\ell\left(-2\log\frac{33}{32}+\log4-\frac34\log3+\frac14\log\frac{15}{8}\right)-2\log6\right]\\
+\frac{p-2}{p-1}\left[-4\log2+\ell\log2\right]-\frac{1}{32}
\end{multlined}\\
&\ge \frac{3}{5}\ell-2\log6-\frac{1}{32}>\left(\frac45\ell-6\right)\log2, \numberthis
\end{align*}
which completes the proof of Theorem~\ref{thm:ext}. \qed

\vspace{2mm}
\noindent
\textbf{Acknowledgments.}
This project used GPT-6 Astra for coming up with the proof method, as well as for general checking and proofreading. We acknowledge support from the U.S.~Department of Energy, Office of Science, Accelerated Research in Quantum Computing, Fundamental Algorithmic Research toward Quantum Utility (FAR-Qu). We were also supported in part by the DoE ASCR Quantum Testbed Pathfinder program (award No.~DE-SC0024220), ONR MURI, NSF QLCI (award No.~OMA-2120757), NSF STAQ program, AFOSR MURI, ARL (W911NF-24-2-0107), and NQVL:QSTD:Design:FTL. We also acknowledge support from the U.S.~Department of Energy, Office of Science, National Quantum Information Science Research Centers, Quantum Systems Accelerator (award No.~DE-SCL0000121). 

\vspace{2mm}
\noindent
\textit{Note added.} 
After posting v1 and during the final stages of preparing v2, we learned of independent work by H.-C. Cheng and P. Wu \cite{cheng2026explicit}, which constructs a similar violation to additivity of minimum output von Neumann entropy.

\bibliographystyle{amsalpha_edit}
\bibliography{moe.bib}

\end{document}